\documentclass[runningheads]{llncs}
\usepackage[T1]{fontenc}
\usepackage{graphicx}
\usepackage{amsmath}
\usepackage{amsfonts}
\usepackage{mathtools}
\usepackage{amssymb}
\usepackage{enumerate}
\usepackage{longtable}
\usepackage{float}
\usepackage{babel}
\usepackage{orcidlink}

\DeclareMathOperator{\vand}{Vand} %extra math operator that I need

\usepackage[
n,
operators,
advantage,
sets,
adversary,
landau,
probability,
notions,	
logic,
ff,
mm,
primitives,
events,
complexity,
asymptotics,
keys]{cryptocode}

\usepackage{hyperref}
\hypersetup{
    hypertexnames = false,% This line solves the problem of hyperlink Appendix sections/tables
    colorlinks   = true, %Colours links instead of ugly boxes
    urlcolor     = blue, %Colour for external hyperlinks
    linkcolor    = blue, %Colour of internal links
    citecolor   = blue, %Colour of citations, could be ``red''
    bookmarksopen=true,
    bookmarksnumbered=true
}

\usepackage[dvipsnames]{xcolor}
\usepackage{bookmark}

\newcommand{\splitatcommas}[1]{%
	\begingroup
	\begingroup\lccode`~=`, \lowercase{\endgroup
		\edef~{\mathchar\the\mathcode`, \penalty0 \noexpand\hspace{0pt plus 1em}}%
	}\mathcode`,="8000 #1%
	\endgroup
}
\DeclareMathOperator{\rank}{rank}

\begin{document}
	
\mainmatter          % for the preliminaries

\title{On the Direct Construction of MDS and Near-MDS Matrices}

\author{~}
%0000-0003-4643-5117
\authorrunning{~}

\institute{~}

\author{Kishan Chand Gupta\inst{1} \and
	Sumit Kumar Pandey\inst{2} \and
	Susanta Samanta\inst{3}}
%0000-0003-4643-5117
\authorrunning{K. C. Gupta et al.}

\institute{Applied Statistics Unit, Indian Statistical Institute, \\ 203, B.T. Road, Kolkata-700108, INDIA. \\ \email{kishan@isical.ac.in} \and Computer Science and Engineering, Indian Institute of Technology Jammu,\\ Jagti, PO Nagrota, Jammu-181221, INDIA. \\ \email{emailpandey@gmail.com} 
\and Department of Electrical and Computer Engineering, University of Waterloo,\\
Waterloo, Ontario, N2L 3G1, Canada \\ 
\email{ssamanta@uwaterloo.ca}}
% %	
\maketitle

\begin{abstract}
	\sloppy
	The optimal branch number of MDS matrices makes them a preferred choice for designing diffusion layers in block ciphers and hash functions. Consequently, various methods have been proposed for designing MDS matrices, including search and direct methods. While exhaustive search is suitable for small-order MDS matrices, direct constructions are preferred for larger orders due to the vast search space involved. In the literature, there has been extensive research on the direct construction of MDS matrices using both recursive and nonrecursive methods.
	On the other hand, in lightweight cryptography, Near-MDS (NMDS) matrices with sub-optimal branch numbers offer a better balance between security and efficiency as a diffusion layer compared to MDS matrices. However, no direct construction method is available in the literature for constructing recursive NMDS matrices.
	This paper introduces some direct constructions of NMDS matrices in both nonrecursive and recursive settings. Additionally, it presents some direct constructions of nonrecursive MDS matrices from the generalized Vandermonde matrices. We propose a method for constructing involutory MDS and NMDS matrices using generalized Vandermonde matrices.  Furthermore, we prove some folklore results that are used in the literature related to the NMDS codes.

    % \keywords{Diffusion Layer, MDS matrix, Near-MDS matrix, Companion matrix, Vandermonde matrix.}

\end{abstract}

\section{Introduction}\label{Sec:Introduction}
% \sloppy
The concept of confusion and diffusion, introduced by Shannon~\cite{SHANON}, is commonly employed in the design of symmetric key cryptographic primitives. Typically, the round function of such designs uses both non-linear and linear layers to achieve confusion and diffusion, respectively. 
The focus of this paper is on the construction of linear diffusion layers that maximize the spreading of internal dependencies.
One way to formalize the concept of perfect diffusion is through the use of multipermutations, which are introduced in \cite{MULT_PER1,MULT_PER2}. Another way to define it is using \textit{Maximum Distance Separable} (MDS) matrices~\cite{AES,AES_2020}. Due to the optimal branch number of MDS matrices, many block ciphers and hash functions use them in their diffusion layers.
\sloppy
In the literature, there has been extensive study of constructing MDS matrices, and we can categorize the approaches of constructing MDS matrices mainly in two ways: nonrecursive and recursive. In nonrecursive constructions, the constructed matrices are themselves MDS. Whereas in recursive constructions, we generally start with a sparse matrix $A$ of order $n$, with a proper choice of elements such that $A^n$ is an MDS matrix.

The advantage of \textit{recursive MDS} matrices is that they are particularly well suited for lightweight implementations: the diffusion layer can be implemented by recursively executing the implementation of the sparse matrix, which requires only a few clock cycles. Recursive MDS matrices based on the companion matrices were used in the PHOTON~\cite{PHOTON} family of hash functions and the LED block cipher~\cite{LED} because companion matrices can be implemented by a simple LFSR.

One can again classify the techniques used to construct MDS matrices based on whether the matrix is constructed directly or found via a search method by enumerating some search space. Exhaustive search works well for small matrices but becomes infeasible for larger orders or over large finite fields due to the rapid growth of the search space. Direct constructions, in contrast, can produce matrices of any order but do not guarantee the lowest implementation cost, which is an important factor in cryptographic applications. The lowest-cost matrices are often obtainable only through search-based methods, but their applicability is limited to small matrices. Direct constructions not only represent a practical option for constructing larger MDS matrices but also provide theoretical insights.

In the literature, there has been extensive research on the direct construction of MDS matrices using both recursive and nonrecursive methods. Nonrecursive direct constructions mainly rely on Cauchy and Vandermonde-based constructions \cite{MDS_Survey,GR13,LACAN2003,LACAN,HMC1,V_MDS}, while recursive direct constructions are obtained through certain coding-theoretic methods. Augot et al.~\cite{Augot2014} employed shortened BCH codes, and Berger~\cite{Berger2013} used Gabidulin codes in their method. Then, in a series of works~\cite{GuptaPV17_2,GuptaPV17_1,GuptaPV19}, the authors proposed many approaches for the construction of recursive MDS matrices from the companion matrices over finite fields.

\textit{Near-MDS} (NMDS) matrices have sub-optimal branch numbers, leading to a slower diffusion speed compared to MDS matrices. However, NMDS matrices can provide a more favorable trade-off between security and efficiency as a diffusion layer, when compared to MDS matrices. Despite their potential benefits, research on NMDS matrices has been limited in the literature, and there is currently no direct construction method available for them in a recursive approach. In 2017, Li et al.~\cite{Li_Wang_2017} explored the construction of NMDS matrices from circulant and Hadamard matrices. In~\cite{Recursive_nMDS_2021}, the focus is on studying the recursive NMDS matrices with the goal of achieving the lowest possible hardware cost. Additionally, recent studies such as~\cite{Huang2021,NMDS_code_2022,NMDS_code_2022_2} have presented direct constructions of NMDS codes, which can be utilized to derive nonrecursive NMDS matrices. In a more recent study~\cite{NMDS_2023}, Gupta et al. explored the construction of NMDS matrices in both recursive and nonrecursive settings and delved into the hardware efficiency of these construction techniques.

\paragraph{\textbf{Contributions:}} As a direct construction, this approach does not guarantee the achievement of lightweight MDS or Near-MDS (NMDS) matrices at the same level attained by the search-based method, whose applicability is constrained by both small matrix dimensions and finite field sizes. Nevertheless, as a novel approach, direct constructions offer valuable alternatives for MDS matrix design. Notably, they provide practical solutions for constructing larger MDS matrices while also delivering significant theoretical insights. This paper introduces several direct constructions for both MDS and NMDS matrices in nonrecursive and recursive frameworks. To clearly highlight our structural novelty, Table~\ref{tab:literature_comparison} provides a comprehensive comparison of our proposed methods against the existing literature. Specifically, this work advances the field through the following primary contributions:
\begin{itemize}
\setlength{\itemsep}{0.6em}
    \item[$\bullet$] We address the lack of direct constructions for recursive NMDS matrices by proposing new constructions based on companion matrices. Additionally, we introduce a new direct construction for recursive MDS matrices.
    \item[$\bullet$] We leverage generalized Vandermonde matrices as a tool for the direct, nonrecursive construction of MDS and NMDS matrices. More specifically, by exploiting their algebraic structure and determinant properties, we provide several direct constructions of both MDS and NMDS matrices over finite fields.
    \item[$\bullet$] Building on this framework, we develop new direct construction methods for involutory MDS matrices and, notably, present the first direct construction of involutory NMDS matrices over finite fields.
    \item[$\bullet$] Finally, we provide formal proofs for several folklore results commonly referenced in the NMDS codes literature.
\end{itemize}

\noindent This paper is structured as follows: Section~\ref{Section_Def_Direct_MDS_NMDS} provides an overview of the mathematical background and notations used throughout, along with proofs of several folklore results on NMDS codes. Section~\ref{Section_Rec_Direct_MDS_NMDS} details direct construction methods for recursive MDS and NMDS matrices, while Section~\ref{Section_nonre_Direct_MDS_NMDS} describes various direct construction methods for nonrecursive MDS and NMDS matrices. Finally, Section~\ref{Section_Conclusion_Direct_Cons} concludes the paper.

\section{Definitions and Preliminaries}\label{Section_Def_Direct_MDS_NMDS}
Let $\mathbb{F}_q$ be the finite field containing $q$ elements, where $q=p^r$ for some prime $p$ and a positive integer $r$. The set of vectors of length $n$ with entries from the finite field $\mathbb{F}_{q}$ is denoted by $\mathbb{F}_{q}^n$. Let $\mathbb{F}_q[x]$ denote the polynomial ring over $\mathbb{F}_q$ in the indeterminate $x$. We denote the algebraic closure of $\mathbb{F}_q$ by $\bar{\mathbb{F}}_q$ and the multiplicative group by $\mathbb{F}_q^{*}$.
It is a well-established fact that elements of a finite field with characteristic $p$ can be represented as vectors with coefficients in $\mathbb{F}_p$. In other words, there exists a vector space isomorphism from $\mathbb{F}_{p^r}$ to $\mathbb{F}_p^r$ defined by $x=(x_1\alpha_1+x_2\alpha_2+ \cdots +x_r\alpha_r) \rightarrow (x_1,x_2, \ldots,x_r)$, where $\{\alpha_1,\alpha_2,\ldots,\alpha_r\}$ is a basis of $\mathbb{F}_{p^r}$ over $\mathbb{F}_{p}$. If $\alpha$ is a primitive element of $\mathbb{F}_{p^r}$, every nonzero element of $\mathbb{F}_{p^r}$ can be expressed as a power of $\alpha$, i.e., $\mathbb{F}_{p^r}^{*}=\set{1,\alpha,\alpha
^2,\alpha^3,\ldots,\alpha^{p^r-2}}$. 

Let ${M}_{k\times n}(\mathbb{F}_q)$ denote the set of all matrices of size $k\times n$ over $\mathbb{F}_q$. For simplicity, we use ${M}_{n}(\mathbb{F}_q)$ to denote the ring of all $n\times n$ matrices (square matrices of order $n$) over $\mathbb{F}_q$. Let $I_{n}$ denote the identity matrix of ${M}_{n}(\mathbb{F}_q)$. The determinant of a matrix $A \in {M}_{n}(\mathbb{F}_q)$ is denoted by $\det(A)$. A square matrix $A$ is said to be nonsingular if  $\det(A)\neq 0$, or equivalently, if the rows (columns) of $A$ are linearly independent over $\mathbb{F}_q$. We now recall some concepts from coding theory. 

A \textit{linear code} $\mathcal{C}$ of length $n$ and dimension $k$ over $\mathbb{F}_q$ is denoted as an $[n,k]$ code. If the minimum distance of $\mathcal{C}$ is equal to $d$ then we denote it as an $[n,k,d]$ code. The \textit{dual code} $\mathcal{C}^{\perp}$ of a code $\mathcal{C}$ can be defined as a subspace of dimension $(n-k)$ that is orthogonal to $\mathcal{C}$. 

A \textit{generator matrix} of $\mathcal{C}$ over $\mathbb{F}_q$ is defined as a $k\times n$ matrix $G$ whose rows form a basis for $\mathcal{C}$. On the other hand, a \textit{parity check matrix} of $\mathcal{C}$ over $\mathbb{F}_q$ is an $(n-k)\times n$ matrix $H$ such that for every $c\in \mathbb{F}_q^n$, $c\in \mathcal{C} \iff Hc^{T}=\mathbf{0}.$ In other words, the code $\mathcal{C}$ is the kernel of $H$ in $\mathbb{F}_q^n$. A generator matrix $G$ is said to be in standard form if it has the form $G=[I_k~|~A]$, where $A$ is a $k \times (n-k)$ matrix. If $G=[I_k~|~A]$ is a generator matrix, then $H=[-A^{T}|I_{n-k}]$ is a parity check matrix for $\mathcal{C}$.

The following lemma establishes a connection between the properties of a parity check matrix and the minimum distance $d$ of a linear code $\mathcal{C}$.

\begin{lemma}\cite[page 33]{FJ77}\label{Lemma_min_dist_v_LI}
	Let $H$ be a parity check matrix of a code $\mathcal{C}$. Then the code has minimum distance $d$ if and only if 
	\begin{enumerate}[(i)]
        \itemsep0em
		\item any $d-1$ columns of $H$ are linearly independent,
		\item some $d$ columns are linearly dependent.
	\end{enumerate}
\end{lemma}
%%%%%%
Constructing a linear code with large values of $k/n$ and $d/n$ is desirable in coding theory. However, there is a trade-off between the parameters $n, k$, and $d$. For instance, the well-known Singleton bound gives an upper bound on the minimum distance for a code.

\begin{theorem}(The Singleton bound)\cite[page 33]{FJ77}
	Let $\mathcal{C}$ be an $[n,k, d ]$ code. Then $d\leq n-k+1$.
\end{theorem}

\begin{definition}(MDS code)
	A code with $d=n - k +1$ is called a maximum distance separable code or an MDS code in short.
\end{definition}

\begin{remark}\label{Remark_MDS_paritycheck}
    An $[n,k]$ MDS code is defined as having minimum distance of $n-k+1$. Thus, every set of $n-k$ columns of the parity check matrix is linearly independent.
\end{remark}

\begin{remark}\label{Remark_MDS_generator}
	Since the dual of an MDS code is again an MDS code~\cite[page 318]{FJ77}, any $k$ columns of the generator matrix are linearly independent.
\end{remark}

Now, we will briefly discuss another important class of linear codes which has found many applications in cryptography. In~\cite{DL95_nMDS}, the concept of Near-MDS codes is introduced as a relaxation of some constraints of the MDS codes. The widely used approach to defining Near-MDS codes is through generalized Hamming weights~\cite{GHamming_Wei91}.

\begin{definition}~\cite{GHamming_Wei91}
    Let $\mathcal{C}$ be an $[n,k]$ code with $\mathcal{D} \subset \mathcal{C}$ as a subcode of $\mathcal{C}$. The support of $\mathcal{D}$, denoted by $\chi(\mathcal{D})$, is the set of coordinate positions, where not all codewords of $\mathcal{D}$ have zero, i.e.,
    \begin{center}
        $\chi(\mathcal{D})=\set{i:\exists (x_1,x_2,\ldots,x_n)\in \mathcal{D}~\text{and}~x_i\neq 0}$.
    \end{center}
\end{definition}

\noindent Using the terminology, an $[n,k]$ code is a linear code of dimension $k$ and support size at most $n$. The rank of a vector space is its dimension, and we may use the terms rank and dimension interchangeably.

\begin{example}\label{Example_GHamming_weights}
    Let $\mathcal{C}$ be the linear code with a generator matrix
    \begin{center}
        $G=
        \begin{bmatrix}
            1 & 0 & 0 & 0 & 1 & 0 \\
            0 & 1 & 0 & 0 & 1 & 1 \\
            0 & 0 & 1 & 0 & 0 & 1
        \end{bmatrix}
        $.
    \end{center}
    Then $\chi(\mathcal{C})=\set{1,2,3,5,6}$ and $\chi(\mathcal{D})=\set{2,3,5,6}$ for the subcode $\mathcal{D}$ generated by the second and third rows of $G$.
\end{example}

\begin{definition}\label{Def_GHamming_weight}\cite{GHamming_Wei91}
    For a linear code $\mathcal{C}$, the $r$-th generalized Hamming weight, denoted as $d_r(\mathcal{C})$, is defined as the cardinality of the minimal support of an $r$-dimensional subcode of $\mathcal{C}$, where $1\leq r \leq k$, i.e.,
	\begin{align*}
		d_r(\mathcal{C})&=\min \set{|\chi(\mathcal{D})|:\mathcal{D}~\text{is a subcode of}~\mathcal{C}~\text{with rank}~r}.
	\end{align*}
\end{definition}

Note that $d_1(\mathcal{C})=d$ is the minimum distance of $\mathcal{C}$.

\begin{example}
    Consider the linear code $\mathcal{C}$ in Example~\ref{Example_GHamming_weights}. It is easy to check that $d_1(\mathcal{C})=2$. By determining the minimal support of all two-dimensional subspaces $\mathcal{D}\subset \mathcal{C}$, we get $d_2(\mathcal{C})=4$. Also, there is at least one codeword in $\mathcal{C}$ with a 1 in each position except the fourth position, which implies that $d_3(\mathcal{C})=5$.
\end{example}

\begin{theorem}(Monotonicity)\label{Th_monotonocity_GHamming}~\cite{GHamming_Wei91}
    For any $[n,k,d]$ code, we have
    \begin{center}
        $1\leq d_1(\mathcal{C})=d< d_2(\mathcal{C})< d_3(\mathcal{C}) \cdots <d_k(\mathcal{C})\leq n$.
    \end{center}
\end{theorem}

\begin{corollary}(Generalized Singleton bound)~\cite{GHamming_Wei91}
    For an $[n,k]$ code $\mathcal{C}$, $d_r(\mathcal{C})\leq n-k+r$. (When $r=1$, this is the Singleton bound.)
\end{corollary}

\noindent Theorem~\ref{Th_Wei91_Th2} provides another method to compute the generalized Hamming weight of a linear code. Let $H$ be a parity check matrix of $\mathcal{C}$ and let $H_i$, $1\leq i \leq n$, be its $i$-th column vector. For any subset of indices $I \subseteq \{1, \dots, n\}$, let $<H_i:i\in I>$ be the space generated by the column vectors $H_i$ for $i\in I$.

\begin{theorem}\label{Th_Wei91_Th2}\cite{GHamming_Wei91}
    For all $r\leq k$,
    \begin{center}
        $d_r(\mathcal{C})=\min \set{|I|:|I|-\rank(<H_i:i\in I>)\geq r}$.
    \end{center}
\end{theorem}

\noindent The following theorem establishes a connection between the properties of a parity check matrix and the generalized Hamming weight of a linear code $\mathcal{C}$. Although this theorem is well-known, we have not found its proof, so we are providing it below.

\begin{theorem}\label{Th_GHamming_paritycheck}\cite{DL95_nMDS,GHamming_Wei91}
    Let $H$ be a parity check matrix for a linear code $\mathcal{C}$. Then $d_r(\mathcal{C})=\delta$ if and only if the following conditions hold:
    \begin{enumerate}[(i)]
        \item any $\delta-1$ columns of $H$ have rank at least $\delta-r$,
        \item there exist $\delta$ columns of $H$ with rank $\delta-r$.
    \end{enumerate}
\end{theorem}

\begin{proof}
    For any $I\subset \set{1,2,\ldots,n}$, let $S(I)=<H_i:i\in I>$ be the space spanned by the vectors $H_i$ for $i\in I$, where $H_i$ denotes the $i$-th column of the parity check matrix $H$ of $\mathcal{C}$. Let 
    \begin{align*}
        S^{\perp}(I)=\set{x\in \mathcal{C}:x_i=0~\text{for}~i \not \in I~\text{and}~ \sum_{i\in I}x_iH_i=\mathbf{0}}.
    \end{align*}
    Then $\rank(S(I))+\rank(S^{\perp}(I))=|I|$.

    Let $d_r(\mathcal{C})=\delta$, and we will prove that both conditions hold. Suppose for contradiction that there exist $\delta-1$ columns of $H$, say $H_{i_1},H_{i_2},\ldots,H_{i_{\delta-1}}$, with rank $\leq \delta-r-1$.

    Now, let $I=\set{i_{1},i_{2},\ldots,i_{\delta-1}}\subset \set{1,2, \ldots, n}$. Then rank$(S(I))\leq \delta-r-1$. Thus, we have
    \begin{align*}
        \rank(S^{\perp}(I))&= |I|-\rank(S(I))\\
        & \geq \delta-1-(\delta-r-1)= r.
    \end{align*}
    Therefore, we have $\rank(S^{\perp}(I))\geq r$. Also, by construction, $S^{\perp}(I)$ is a subcode of $\mathcal{C}$ and $|\chi(S^{\perp}(I))|\leq \delta-1$. This leads to a contradiction since $d_r(\mathcal{C})=\delta$. Therefore, we can conclude that any $\delta-1$ columns of $H$ have rank greater than or equal to $\delta-r$.

    Since $d_r(\mathcal{C})=\delta$, there exists a subcode $\mathcal{D}$ of $\mathcal{C}$ with $\rank(\mathcal{D})=r$ and $|\chi(\mathcal{D})|=d_r(\mathcal{C})$. Let $I=\chi(\mathcal{D})$. Now, we will show that $\mathcal{D}=S^{\perp}(I)$.

    Let $c=(c_1,c_2,\ldots,c_n)\in \mathcal{D}$ be a codeword. Then we have
    \begin{align*}
        & \sum_{i=1}^{n}c_iH_i=\mathbf{0}\\
        \implies & \sum_{i\in I}c_iH_i+\sum_{i \not \in I}c_iH_i=\mathbf{0}\\
        \implies & \sum_{i\in I}c_iH_i=\mathbf{0}~~ [\text{Since}~c_i=0 ~\forall i \not \in I=\chi(\mathcal{D})]\\
        \implies & c\in S^{\perp}(I)\\
        \implies & \mathcal{D}\subset S^{\perp}(I).
    \end{align*}

    If possible, let $\rank(S^{\perp}(I))=r'>r$. Now, since $\rank(S(I))+\rank(S^{\perp}(I))=|I|$, we have 
    \begin{align*}
        &|I|-\rank(S(I))=r'>r\\
        \implies & d_{r'}(\mathcal{C})\leq |I|=\delta ~~ [\text{By Theorem~\ref{Th_Wei91_Th2}}].
    \end{align*}
    But by the monotonicity of generalized Hamming weights we must have 
    \begin{align*}
        \delta=d_{r}(\mathcal{C})<d_{r'}(\mathcal{C})\leq \delta,
    \end{align*}
    which is a contradiction. Hence, we must have $\rank(\mathcal{D})=\rank(S^{\perp}(I))$ and $\mathcal{D}=S^{\perp}(I)$. Thus, 
    \begin{align*}
        &\rank(S(I))=|I|-r=\delta-r.
    \end{align*}
    Therefore, there exist $\delta$ columns in $H$ of rank $\delta-r$.

    For the converse part, assume that both conditions hold. From condition $(ii)$, we know that there exist some $I\subset \set{1,2,\ldots,n}$ with $|I|=\delta$ such that $\rank(S(I))=\delta-r$. This implies that

    \begin{center}
        $\rank(S^{\perp}(I))=|I|-\rank(S(I))=r$.
    \end{center}

    Since $|I|-\rank(S(I))=r$, by Theorem~\ref{Th_Wei91_Th2}, we have $d_r(\mathcal{C})\leq \delta$.

    If possible, let $d_r(\mathcal{C})=\delta-t$ for some $t\geq 1$. Now, by Theorem~\ref{Th_Wei91_Th2}, there exist some $I'\subset \set{1,2,\ldots,n}$ with $|I'|=\delta-t$ such that
    \begin{align*}
        &|I'|-\rank(S(I'))\geq r\\
        \implies & \rank(S(I'))\leq |I'|-r\\
        \implies & \rank(S(I'))\leq \delta-t-r.
    \end{align*}
    Therefore, there exist $|I'|=\delta-t$ columns, say $H_{i_1},H_{i_2},\ldots,H_{i_{\delta-t}}$, of $H$ of rank $\leq \delta-t-r$. Now, by adding any other $t-1$ columns of $H$ to those $\delta-t$ columns we have $\delta-1$ columns, say $H_{i_1},H_{i_2},\ldots,H_{i_{\delta-t}},H_{i_{\delta-t+1}},\ldots,H_{i_{\delta-1}}$, of $H$ of rank $\leq (\delta-t-r)+(t-1)=\delta-r-1<\delta-r$. This contradicts condition $(i)$. Hence, we must have $d_r(\mathcal{C})=\delta$.\qed
\end{proof}

\begin{definition}(NMDS code)\label{Def_NMDS_code}\cite{DL95_nMDS}
    An $[n, k]$ code $\mathcal{C}$ is said to be Near-MDS or NMDS if 
    \begin{align*}
        &d_1(\mathcal{C})=n-k~~\text{and}~~d_i(\mathcal{C})=n-k+i,~~\text{for}~i=2,3,\ldots,k.
    \end{align*}
\end{definition}

\begin{remark}
    From the monotonicity of generalized Hamming weights, we can say that an $[n, k]$ code is NMDS if and only if $d_1(\mathcal{C})=n-k~\text{and}~d_2(\mathcal{C})=n-k+2$.
\end{remark}

\begin{remark}\label{Remark_AMDS_is_not_NMDS}
	For an $[n,k,d]$ code $\mathcal{C}$, if $d=n-k$, then $\mathcal{C}$ is called an Almost-MDS or AMDS code. However, it is worth noting that not all AMDS codes are necessarily NMDS codes. For example, consider the linear code $\mathcal{C}$ with a generator matrix
    \begin{align*}
        G&=\begin{bmatrix}
            1 & 0 & 0 & \alpha^2 & \alpha & 0 \\
            0 & 1 & 0 & \alpha & \alpha & 0 \\
            0 & 0 & 1 & \alpha & 0 & \alpha
        \end{bmatrix}
    \end{align*}
    over the finite field $\mathbb{F}_{2^2}$ constructed by the polynomial $x^2 + x + 1$ and $\alpha$ is a root of $x^2 + x + 1$. Then it can be checked that $\mathcal{C}$ is a $[6,3,3]$ code. Also, by determining the minimal support of all two-dimensional subspaces $\mathcal{D}\subset \mathcal{C}$, we get $d_2(\mathcal{C})=4<5$. This value is achieved by the subspace spanned by the first two rows of the generator matrix $G$. Hence, $\mathcal{C}$ is not an NMDS code. 
	
	However, when both $\mathcal{C}$ and its dual $\mathcal{C}^{\perp}$ are AMDS codes, then $\mathcal{C}$ is classified as an NMDS code~\cite{Boer1996_AMDS}.
\end{remark}

\noindent Theorem~\ref{Th_GHamming_paritycheck} provides the following useful result on NMDS codes.

\begin{lemma}\label{Lemma_three_conditions_NMDS}\cite{DL95_nMDS}
    Let $H$ be a parity check matrix of an $[n,k]$ code $\mathcal{C}$. Then the code $\mathcal{C}$ is NMDS if and only if $H$ satisfies the conditions
	\begin{enumerate}[(i)]
		\item any $n-k-1$ columns of $H$ are linearly independent,
		\item there exist some $n-k$ columns that are linearly dependent,
		\item any $n-k+1$ columns of $H$ are of full rank.
	\end{enumerate}
\end{lemma}

\begin{proof}
    Let $\mathcal{C}$ be an NMDS code. Therefore, we have $d_1=n-k$ and $d_2=n-k+2$. Since $d_1$ is the minimum distance of $\mathcal{C}$, from Lemma~\ref{Lemma_min_dist_v_LI}, we can say that $d_1=n-k$ if and only if any $n-k-1$ columns of $H$ are linearly independent and there exist some $n-k$ columns that are linearly dependent. Moreover, Theorem~\ref{Th_GHamming_paritycheck} implies that $d_2=n-k+2$ if and only if any $n-k+1$ columns of $H$ have rank greater than or equal to $(n-k+2)-2=n-k$ and there exist $n-k+2$ columns of $H$ of rank $(n-k+2)-2=n-k$. Since $H$ is a parity check matrix of $\mathcal{C}$, we have $\rank(H)=n-k$. Therefore, we can conclude that $d_2=n-k+2$ if and only if any $n-k+1$ columns of $H$ are of full rank. This completes the proof.\qed
\end{proof}

\noindent It can be deduced from the properties of the generalized Hamming weights that the dual of an NMDS code is also an NMDS code.

\begin{lemma}\label{Lemma_dual_NMDS}\cite{DL95_nMDS}
    If an $[n,k]$ code is NMDS, then its dual code is also NMDS.
\end{lemma}

One can infer from Lemma~\ref{Lemma_dual_NMDS} that a generator matrix of an $[n,k]$ NMDS code must satisfies conditions similar to those in Lemma~\ref{Lemma_three_conditions_NMDS}.

\begin{lemma}\label{Lemma_generator_matrix_NMDS}\cite{DL95_nMDS}
    Let $G$ be a generator matrix of an $[n,k]$ code $\mathcal{C}$. Then the code $\mathcal{C}$ is NMDS if and only if $G$ satisfies the conditions
	\begin{enumerate}[(i)]
		\item any $k-1$ columns of $G$ are linearly independent,
		\item there exist $k$ columns that are linearly dependent,
		\item any $k+1$ columns of $G$ have full rank.
	\end{enumerate}
\end{lemma}

\noindent We will now explore MDS and NMDS matrices, which have notable cryptographic applications. The concept of MDS and NMDS matrices is derived from the MDS and NMDS codes, respectively. Generally, the matrix $A$ in the generator matrix $G=[I~|~A]$ of an $[n,k]$ code $\mathcal{C}$ is considered to be an MDS or NMDS matrix depending on whether the code $\mathcal{C}$ is MDS or NMDS. Since square matrices are typically used in practice, for the sake of simplicity, we will consider the $[2n,n]$ code instead of the generic form of the $[n,k]$ code throughout the rest of this paper.

\begin{definition}\label{NMDS_DEF}
	A matrix $A$ of order $n$ is said to be an MDS (NMDS) matrix if $[I~|~A]$ is a generator matrix of a $[2n,n]$ MDS (NMDS) code.
\end{definition}
Since the dual of an MDS code is also an MDS code, and Lemma~\ref{Lemma_dual_NMDS} demonstrates that the dual of an NMDS code is an NMDS code, we can consequently deduce the following results regarding MDS and NMDS matrices.

\begin{corollary}\label{transpose_is_NMDS}
	If $A$ is an MDS (NMDS) matrix, then $A^{T}$ is also an MDS (NMDS) matrix.
\end{corollary}

The goal of lightweight cryptography is to design ciphers that require minimal hardware resources, consume low energy, exhibit low latency, and optimize their combinations. One proposed method for reducing chip area is the use of recursive MDS (NMDS) matrices.

\begin{definition}
    Let $s$ be a positive integer. We say that a matrix $B$ is recursive MDS (NMDS) or $s$-MDS ($s$-NMDS) if the matrix $A=B^s$ is MDS (NMDS). If $B$ is $s$-MDS ($s$-NMDS), then we say that $B$ yields an MDS (NMDS) matrix.
\end{definition}

\begin{example}\label{Example_22-MDS}
    The matrix	
    $$B=
    \left[ \begin{array}{rrrr}
    0 & 1 & 0 & 0 \\
    0 & 0 & 1 & 0 \\
    0 & 0 & 0 & 1 \\
    1 & \alpha & 0 & 0
    \end{array} \right]
    $$
    is $22$-MDS and 10-NMDS, where $\alpha$ is a primitive element of the field $\FF_{2^4}$ and a root of $x^4+x+1$.
\end{example}
%%%

\noindent Vandermonde matrices have gained significant attention in the literature of constructing MDS codes. However, Vandermonde matrices defined over a finite field may contain singular square submatrices~\cite[Page~323]{FJ77}. Consequently, these matrices by themselves need not be MDS. To address this issue, Lacan and Fimes \cite{LACAN2003,LACAN} employed two Vandermonde matrices to construct an MDS matrix. Later, Sajadieh et al. \cite{V_MDS} used a similar approach to obtain an MDS matrix that is also involutory.

\begin{definition}(Vandermonde matrix){\label{def:vandermonde}}
	The matrix
	\begin{center}
		$A= \vand(x_{1},x_{2},\ldots,x_{n})=
		\begin{bmatrix}
            1 ~ & ~ 1 ~ & ~  \ldots ~ & ~ 1 \\
            x_{1} ~ & ~ x_{2} ~ & ~ \ldots ~ & ~ x_{n} \\
            x_{1}^2 ~ & ~ x_{2}^2 ~  & ~ \ldots ~ & ~ x_{n}^2 \\
            \vdots ~ & ~ \vdots ~ & ~ \vdots ~ &  ~ \vdots \\
            x_{1}^{n-1} ~ & ~ x_{2}^{n-1} ~ & ~ \ldots ~ & ~ x_{n}^{n-1} \\
		\end{bmatrix}$
	\end{center}
	is called a Vandermonde matrix, where $x_{i}$'s are elements of a finite or infinite field.
\end{definition}

We sometimes use the notation $\vand({\bf x})$ to represent the Vandermonde matrix $\vand(\splitatcommas{x_{1},x_{2},\ldots,x_{n}})$, where ${\bf x}=(x_{1},x_{2},\ldots,x_{n})$. It is known that 
\begin{center}
    $\det(\vand({\bf x}))= \displaystyle{\prod_{1\leq i < j \leq n} (x_{j} - x_{i})}$,
\end{center}
which is nonzero if and only if the $x_{i}$'s are distinct.

\noindent There are several generalizations of the Vandermonde matrices in the literature, as documented in~\cite{Gvand_2,GVand_3,GVand,GVand_4,GVand_5,GVand_6} and the references therein. Our focus is on the variant presented in~\cite{GVand}, due to its applications in cryptography and error correcting codes. The definition of this variant is as follows.

\begin{definition}(Generalized Vandermonde matrix)\label{Def_GVand_intro}
    Let ${\bf x} = (x_1,x_2,\ldots,x_n) \in \mathbb{F}_q^n$ and let $I$ be a finite set of nonnegative integers. Let $t_1 < t_2 < \ldots < t_n$ be the ordered elements of the set $\{0, 1, \ldots, n + |I| - 1\} \setminus I$. Then the matrix
    \[ V_{\perp}({\bf x};I) = \left[
    \begin{array}{cccc}
    x_1^{t_1} ~ & ~ x_2^{t_1} ~ & ~ \ldots ~ & ~ x_n^{t_1}\\
    x_1^{t_2} ~ & ~ x_2^{t_2} ~ & ~ \ldots ~ & ~ x_n^{t_2}\\
    \vdots & \vdots & \vdots \\
    x_1^{t_n} ~ & ~ x_2^{t_n} ~ & ~ \ldots ~ & ~ x_n^{t_n}\\
    \end{array}
    \right]
    \]
    is said to be a generalized Vandermonde matrix with respect to $I$.
\end{definition}

\begin{remark}
    Observe that the matrix $V_{\perp}({\bf x};I)$ is the standard Vandermonde matrix $\vand({\bf x})$ if $I = \varnothing$, in which case the powers are simply $0, 1, \ldots, n-1$.
\end{remark}

Now, we will see how the determinant of $V_{\perp}({\bf x};I)$ can be computed with the help of the determinant of the Vandermonde matrix $\vand({\bf x})$ when $I$ is nonempty. To do so, we require the following definition.

\begin{definition}
    The elementary symmetric polynomial of degree $d$ is defined as
    \begin{align*}
        \sigma_{d}(x_1,x_2,\ldots,x_n)&=\sum_{w(e)=d}^{}{x_1^{e_1}x_2^{e_2}\cdots x_n^{e_n}},
    \end{align*}
    where $e=(e_1,e_2,\ldots,e_n)\in \mathbb{F}_2^n$.
\end{definition}

\begin{theorem}\cite[Theorem~1]{GVand}\label{Th_GVand_discon}
    If $I=\set{l_1,l_2,\ldots,l_s}$, we have
    \begin{center}
        $\det(V_{\perp}({\bf x}; I))=\det(\vand({\bf x}))\det(S({\bf x}))$,
    \end{center}
    where $S({\bf x})$ is the $s \times s$ matrix defined as
    \[S({\bf x}) = \begin{bmatrix}
    \sigma_{n-l_1}({\bf x}) & \sigma_{n-l_1+1}({\bf x}) & \cdots & \sigma_{n-l_1+s-1}({\bf x}) \\
    \sigma_{n-l_2}({\bf x}) & \sigma_{n-l_2+1}({\bf x}) & \cdots & \sigma_{n-l_2+s-1}({\bf x}) \\
    \vdots & \vdots & \ddots & \vdots \\
    \sigma_{n-l_s}({\bf x}) & \sigma_{n-l_s+1}({\bf x}) & \cdots & \sigma_{n-l_s+s-1}({\bf x})
    \end{bmatrix}. \]
\end{theorem}

\begin{lemma}\cite[Lemma~1]{GVand}\label{Lemma_GVand_Det}
    If $I=\set{l}$, we have
    \begin{center}
        $\det(V_{\perp}({\bf x}; I))=\det(\vand({\bf x}))\sigma_{n-l}({\bf x})$.
    \end{center}
\end{lemma}

By substituting $I=\set{n-1}$ and $I=\set{1}$ into Lemma~\ref{Lemma_GVand_Det}, we can derive Corollaries~\ref{Corollary_GVand_det1} and \ref{Corollary_GVand_det2}, respectively.

\begin{corollary}\label{Corollary_GVand_det1}
	Let $I=\set{n-1}$, then $\det(V_{\perp}({\bf x}; I))=\det(\vand({\bf x}))(\sum_{i=1}^{n}{x_i})$.
\end{corollary}

\begin{corollary}\label{Corollary_GVand_det2}
	Let $I=\set{1}$ and each $x_i$ be a nonzero element of a field. Then we can express the determinant of $V_{\perp}({\bf x}; I)$ as
	\begin{center}
		$\det(V_{\perp}({\bf x}; I))=(\prod_{i=1}^{n}x_i) \det(\vand({\bf x}))(\sum_{i=1}^{n}x_i^{-1})$.
	\end{center}
\end{corollary}

Now, we will consider the case when $I$ has more than one element, specifically, we will explore how to compute the determinant of $V_{\perp}({\bf x}; I)$ when $I=\set{1,n}$.

\begin{corollary}\label{Corollary_GVand_det3}
	Let $I=\set{1,n}$ and each $x_i$ be a nonzero element of a field. Then we can express the determinant of $V_{\perp}({\bf x}; I)$ as
	\begin{equation*}
		\det(V_{\perp}({\bf x}; I))=\det(\vand({\bf x}))\left(\prod_{i=1}^{n}x_i \right) \left[ (\sum_{i=1}^{n}x_i)(\sum_{i=1}^{n}x_i^{-1})-1 \right].
	\end{equation*}
\end{corollary}

\begin{proof}
	From Theorem~\ref{Th_GVand_discon}, we know that \begin{center}
        $\det(V_{\perp}({\bf x}; I))=\det(\vand({\bf x}))\det(S({\bf x}))$,
    \end{center}
	where 
	$S({\bf x})=\begin{bmatrix}
		\sigma_{n-1}({\bf x}) &~ \sigma_{n}({\bf x})\\
		\sigma_{0}({\bf x}) &~ \sigma_{1}({\bf x})
	\end{bmatrix}$. Thus, we have
	\begin{equation*}
		\begin{aligned}
			\det(S({\bf x}))
			&= \sigma_{n-1}({\bf x}) \sigma_{1}({\bf x}) - \sigma_{n}({\bf x}) \sigma_{0}({\bf x})\\
			&= \left[(\prod_{i=1}^{n}x_i \sum_{i=1}^{n}x_i^{-1})(\sum_{i=1}^{n}x_i)\right] - \prod_{i=1}^{n}x_i\\
			&= \prod_{i=1}^{n}x_i \left[(\sum_{i=1}^{n}x_i)(\sum_{i=1}^{n}x_i^{-1})-1 \right].
		\end{aligned}
	\end{equation*}
	Therefore, $\det(V_{\perp}({\bf x}; I))=\det(\vand({\bf x})) \left(\prod_{i=1}^{n}x_i \right) \left[(\sum_{i=1}^{n}x_i)(\sum_{i=1}^{n}x_i^{-1})-1\right]$.\qed
\end{proof}

Now, let us recall the companion matrix structures which are used for the construction of recursive MDS matrices.

\begin{definition}(Companion matrix)\label{def_companion}
	Let $g(x) = a_1 + a_2 x + \ldots + a_{n} x^{n-1} + x^n\in \mathbb{F}_q[x]$ be a monic polynomial of degree $n$. The companion matrix $C_g\in M_n(\mathbb{F}_{q})$ associated with the polynomial $g(x)$ is given by
	\[ C_g = \left[ \begin{array}{ccccc}
	0         & 1    & 0          &\ldots   & 0 \\
	\vdots &       & \ddots &             & \vdots\\
	0         & 0    &  \ldots   & \ldots  & 1 \\
	-a_1     & -a_2 & \ldots  &  \ldots  & -a_{n}
	\end{array} \right].\]
\end{definition}

\begin{definition}
	A square matrix $M\in M_n(\mathbb{F}_{q})$ is said to be diagonalizable if $M$ is similar to a diagonal matrix. This means $M= PDP^{-1}$ for some diagonal matrix $D$ and a nonsingular matrix $P$.
\end{definition}

Now, we will consider some results related to diagonalizable companion matrices.

\begin{lemma}\cite{GuptaPV15}\label{lemma_Cg=VDV^{-1}}
Let $C_g \in M_n(\mathbb{F}_{q})$ be a nonsingular
companion matrix which is diagonalizable, say $C_g= P DP^{-1}$ where P is a nonsingular matrix of order $n$ and $D = diag(\lambda_1,\lambda_2,\ldots,\lambda_n)$. Then all entries of $P$ are nonzero. Moreover, $C_g$ can be expressed as $C_g=VDV^{-1}$, where $V = \vand(\lambda_1,\lambda_2,\ldots,\lambda_n)$.
\end{lemma}

\begin{corollary}\cite{GuptaPV15}
	A companion matrix $C_g$ is nonsingular and diagonalizable if and only if all eigenvalues of $C_g$ are distinct and nonzero.
\end{corollary}

\begin{lemma}\cite{RAO}\label{Lemma_distinct_eigenvalue}
	If $M$ is an $n\times n$ matrix with $n$ distinct eigenvalues, then $M$ is diagonalizable.
\end{lemma}

\begin{theorem}\cite{RAO}\label{Th_char_poly_companion}
	The characteristic polynomial of $C_g$, as defined in Definition~\ref{def_companion}, is the polynomial $g(x)=a_1 + a_2 x + \ldots + a_{n} x^{n-1} + x^n$.
\end{theorem}

Since the roots of a characteristic polynomial are the eigenvalues, based on Lemma~\ref{lemma_Cg=VDV^{-1}}, Lemma~\ref{Lemma_distinct_eigenvalue}, and Theorem~\ref{Th_char_poly_companion}, we can conclude the following result for a companion matrix.

\begin{theorem}\label{Th_C=VDV^{-1}}
	If the monic polynomial $g(x)=a_1 + a_2 x + \ldots + a_{n} x^{n-1} + x^n$ has $n$ distinct nonzero roots $\lambda_1,\lambda_2,\ldots,\lambda_n$, then $C_g$ can be expressed as $C_g=VDV^{-1}$, where $V = \vand(\lambda_1,\lambda_2,\ldots,\lambda_n)$ and $D=diag(\lambda_1,\lambda_2,\ldots,\lambda_n)$.
\end{theorem}

\noindent We know that the rows of the generator matrix $G$ form a basis of an $[\ell,n]$ linear code $\mathcal{C}$ with $\operatorname{rank}(G)=n$. We also know that if $A$ is a nonsingular matrix, then $\operatorname{rank}(AG)=\operatorname{rank}(G)=n$. Hence, the rows of $AG$ are linearly independent and span $\mathcal{C}$, so $AG$ is another generator matrix of $\mathcal{C}$. Thus, we have the following lemma.

\begin{lemma}\label{Lemma_AG_another_gen_matrix}
    Let $A$ be an $n\times n$ nonsingular matrix and $G$ be a generator matrix of an $[\ell,n]$ code $\mathcal{C}$. Then $AG$ is also a generator matrix of the code $\mathcal{C}$.
\end{lemma}

\noindent \textbf{Remark on Index Notation:} In the subsequent sections, multiple index sets are occasionally utilized simultaneously to define the parameters and submatrices of our constructions. To prevent ambiguity, we establish the following notations: $E$ represents the full set of available indices in a given context (for example, $E = \{0, 1, \dots, n-1, m, \dots, m+n-1\}$). The variable $R$ strictly denotes a specific subset of indices (typically $R \subset E$) corresponding to the columns currently being evaluated for linear independence. Finally, in Section~\ref{Section_nonre_Direct_MDS_NMDS}, $I$ is strictly reserved to denote the set of powers within the generalized Vandermonde matrices $V_{\perp}({\bf x}; I)$.

\section{Direct Construction of Recursive MDS and NMDS Matrices}\label{Section_Rec_Direct_MDS_NMDS}
\sloppy
In this section, we present various techniques for direct construction of MDS and NMDS matrices over finite fields, in recursive approach. To the best of our knowledge, we are the first to provide a direct construction method for recursive NMDS matrices.
We begin by establishing a condition for the similarity between a companion matrix and a diagonal matrix. Using this condition, we can represent the companion matrix as a combination of a Vandermonde matrix and a diagonal matrix. We utilize determinant expressions for generalized Vandermonde matrices to present several techniques for constructing recursive NMDS matrices that are derived from companion matrices. Furthermore, a new direct construction for recursive MDS matrices is introduced.

\begin{lemma}\label{Lemma_G'_another_gen_matrix}
	Let $g(x)\in \mathbb{F}_q[x]$ be a monic polynomial of degree $n$ with $n$ distinct roots, say $\lambda_1,\ldots,\lambda_n\in\bar{\mathbb{F}}_q$. Then the matrix 
	\begin{equation}\label{Eqn_RNMDS_G}
		\begin{aligned}
			G' =\left[ \begin{array}{cccccccc}
				1 & \lambda_1 & \ldots & \lambda_1^{n-1} & \lambda_1^m & \lambda_1^{m+1} & \ldots & \lambda_1^{m+n-1}\\
				\vdots & \vdots & \ddots & \vdots & \vdots & \vdots & \ddots & \vdots\\
				1 & \lambda_n & \ldots & \lambda_n^{n-1} & \lambda_n^m & \lambda_n^{m+1} & \ldots & \lambda_n^{m+n-1}\\
				\end{array} \right]
		\end{aligned}
	\end{equation}
	is also a generator matrix for the $[2n,n]$ code $\mathcal{C}$ with generator matrix $G=[I~|~(C_g^T)^m]$.
\end{lemma}

\begin{proof}
From Theorem~\ref{Th_C=VDV^{-1}}, we know that if a polynomial $g(x)$ has $n$ distinct roots $\lambda_1, \ldots, \lambda_{n}$, then the companion matrix $C_g$ associated to $g(x)$ can be written as $C_g=VDV^{-1}$, where 
\begin{equation*}
	\begin{aligned}
		V&=\vand(\lambda_{1},\lambda_{2},\ldots,\lambda_{n})\\
		&=
		\begin{bmatrix}
            1 ~ & ~ 1 ~ & ~ \ldots ~ & ~ 1 \\
            \lambda_{1} ~ & ~ \lambda_{2} ~ & ~ \ldots ~ & ~ \lambda_{n} \\
            \lambda_{1}^2 ~ & ~ \lambda_{2}^2 ~ & ~ \ldots ~ & ~ \lambda_{n}^2 \\
            \vdots ~ & ~ \vdots ~ & ~ \vdots ~ &  ~ \vdots \\
            \lambda_{1}^{n-1} ~ & ~ \lambda_{2}^{n-1} ~ & ~ \ldots ~ & ~ \lambda_{n}^{n-1} \\
		\end{bmatrix}
	\end{aligned}
\end{equation*}
and $D=diag(\lambda_1, \ldots, \lambda_{n})$.

Let $\mathcal{C}$ be a $[2n,n]$ code with generator matrix $G=[I~|~(C_g^T)^m]$. Now

\begin{equation}\label{Eqn_G'=AG}
	\begin{aligned}
		G&= [I ~|~ (C_g^T)^m]=[I ~|~ ((V^T)^{-1}DV^T)^m]\\
		 &= [I ~|~ (V^T)^{-1}D^mV^T]\\
		 &= (V^T)^{-1}[V^T ~|~ D^m V^T]\\
		 &=(V^T)^{-1} G',
	\end{aligned}
\end{equation}
where $G' = [V^T ~|~ D^m V^T]$.
% \[
% G=[I ~|~ (C_g^T)^m] = [I ~|~ (V^T)^{-1}D^mV^T] = (V^T)^{-1}[V^T ~|~ D^m V^T] = (V^T)^{-1} G'
% \]
% where $G' = [V^T ~|~ D^m V^T]$. Therefore, we have 
% from Lemma~\ref{Lemma_AG_another_gen_matrix}
Therefore, we have
\begin{equation*}
	\begin{aligned}
		G' & = [V^T ~|~ D^m V^T] \\
		   & =\left[ \begin{array}{cccccccc}
			1 & \lambda_1 & \ldots & \lambda_1^{n-1} & \lambda_1^m & \lambda_1^{m+1} & \ldots & \lambda_1^{m+n-1}\\
			\vdots & \vdots & \ddots & \vdots & \vdots & \vdots & \ddots & \vdots\\
			1 & \lambda_n & \ldots & \lambda_n^{n-1} & \lambda_n^m & \lambda_n^{m+1} & \ldots & \lambda_n^{m+n-1}\\
			\end{array} \right].
	\end{aligned}
\end{equation*}
Also, from~(\ref{Eqn_G'=AG}), we have $G'=V^T G$. Hence, according to Lemma~\ref{Lemma_AG_another_gen_matrix}, we can conclude that $G'$ is also a generator matrix for the linear code $\mathcal{C}$.\qed
\end{proof}
% For future references, we summarize the above discussion in the following two theorems.

Let $C_g$ be the companion matrix associated with a monic polynomial $g(x)$ of degree $n\geq 3$. Then for $m<n$, it can be observed that the first row of $C_g^m$ is a unit vector. Hence, the linear code generated by $[I~|~C_g^m]$ has minimum distance less than $n$. Therefore, for $m<n$, $C_g^m$ cannot be an MDS or NMDS matrix.

\begin{theorem}\label{Th_RMDS_MainResult}
	Let $g(x)\in \mathbb{F}_q[x]$ be a monic polynomial of degree $n$. Suppose that $g(x)$ has $n$ distinct roots, say $\lambda_1,\ldots,\lambda_n\in\bar{\mathbb{F}}_q$. Let $m$ be an integer with $m\geq n$. Then the matrix $M = C_g^m$ is MDS if and only if any $n$ columns of the matrix $G'$ given in~(\ref{Eqn_RNMDS_G}) are linearly independent.
\end{theorem}

\begin{proof}
	From Corollary~\ref{transpose_is_NMDS}, we know that $C_g^m$ is an MDS matrix if and only if its transpose $(C_g^m)^T = (C_g^T)^m$ is also an MDS matrix. Also, according to Definition~\ref{NMDS_DEF}, $(C_g^T)^m$ is MDS if and only if the $[2n,n]$ code $\mathcal{C}$, with generator matrix $G=[I~|~(C_g^T)^m]$, is an MDS code. 
	
	Now, since $\lambda_1,\ldots,\lambda_n$ are $n$ distinct roots of $g(x)$, from Lemma~\ref{Lemma_G'_another_gen_matrix}, we can say that the matrix $G'$ in~(\ref{Eqn_RNMDS_G}) is also a generator matrix for the code $\mathcal{C}$. Therefore, by Remark~\ref{Remark_MDS_generator}, we can establish that $(C_g^m)^T$ is MDS, and hence $C_g^m$, if and only if any $n$ columns of $G'$ are linearly independent.\qed
\end{proof}

\begin{theorem}\label{Th_RNMDS_MainResult}
	Let $g(x)\in \mathbb{F}_q[x]$ be a monic polynomial of degree $n$. Suppose that $g(x)$ has $n$ distinct roots, say $\lambda_1,\ldots,\lambda_n\in\bar{\mathbb{F}}_q$. Let $m$ be an integer with $m\geq n$. Then the matrix $M = C_g^m$ is NMDS if and only if the matrix $G'$ given in~(\ref{Eqn_RNMDS_G}) satisfies the three conditions in Lemma~\ref{Lemma_generator_matrix_NMDS}.
\end{theorem}

\begin{proof}
	From Corollary~\ref{transpose_is_NMDS}, we know that $C_g^m$ is an NMDS matrix if and only if its transpose $(C_g^m)^T = (C_g^T)^m$ is also an NMDS matrix. Also, by Definition~\ref{NMDS_DEF}, $(C_g^T)^m$ is an NMDS matrix if and only if the $[2n,n]$ code $\mathcal{C}$, with generator matrix $G=[I~|~(C_g^T)^m]$, is an NMDS code. 
	
	As $\lambda_1,\ldots,\lambda_n$ are $n$ distinct roots of $g(x)$, we can infer from Lemma~\ref{Lemma_G'_another_gen_matrix} that the matrix $G'$ defined in~(\ref{Eqn_RNMDS_G}) is also a generator matrix for the code $\mathcal{C}$. Consequently, we can conclude that $(C_g^m)^T$ is NMDS, and hence $C_g^m$ is NMDS, if and only if the matrix $G'$ satisfies the three conditions in Lemma~\ref{Lemma_generator_matrix_NMDS}.\qed
\end{proof}
 
% It would be valuable to explore other methods for generating such polynomials.

\noindent 
Now, we present two methods for the construction of polynomials that yield recursive NMDS matrices. The polynomials constructed using these methods have distinct roots. The main idea behind these methods is Theorem~\ref{Th_RNMDS_MainResult}: we suitably choose $\lambda_i$, for $1\le i\le n$, and verify that the polynomial $g(x) = \prod_{i=1}^n (x-\lambda_i) \in \mathbb{F}_q[x]$ satisfies the condition of Theorem~\ref{Th_RNMDS_MainResult}. 
To do so, we must examine the rank of the submatrices of $G'$ constructed from any $t$ columns (here we examine $t=n-1,n,n+1$) of $G'$ corresponding to $\lambda_i$'s as given in~(\ref{Eqn_RNMDS_G}).
A submatrix $G'[R]$, constructed from any $t$ columns of $G'$, is given by

\begin{equation}\label{Eqn_RNMDS_submatrix_G'}
	G'[R] =
	\left[ \begin{array}{cccccccc}
	\lambda_1^{r_1} & \lambda_1^{r_2} & \ldots & \lambda_1^{r_t}\\
	\lambda_2^{r_1} & \lambda_2^{r_2} & \ldots & \lambda_2^{r_t}\\
	\vdots & \vdots & \ddots & \vdots \\
	\lambda_n^{r_1} & \lambda_n^{r_2} & \ldots & \lambda_n^{r_t}
	\end{array} \right],
\end{equation}
where $R$ denotes a set $\{r_1,r_2,\ldots,r_t\} \subset E = \{0,1,\ldots,n-1,m,m+1,\ldots,m+n-1\}$ of $t$ elements.

\noindent Before detailing the formal algebraic proofs, we provide a brief conceptual intuition for the following theorems. To transition from a MDS matrix to a NMDS matrix, we must deliberately relax the optimal branch number. In the framework of generalized Hamming weights, this requires intentionally inducing specific, controlled linear dependencies (rank deficiencies) among certain column subsets. In Theorem \ref{Th_RNMDS_ConsI(b)}, we achieve this by selecting the roots $\lambda_i$ of the companion matrix's polynomial. By enforcing explicit zero-sum constraints on these roots (e.g., $\sum_{i=1}^{n}\theta^{r_{i}}=0$), we artificially create the exact linear dependencies required by the NMDS criteria, while mathematically guaranteeing that all other submatrices remain full rank.

\begin{theorem}\label{Th_RNMDS_ConsI(b)}
	Let $\lambda_i = \theta^{i-1}$ for $1\le i\le n-1$ and $\lambda_n = \theta^n$ for some $\theta\in \mathbb{F}_q^{*}$.
	Let $g(x) = \prod_{i=1}^n (x-\lambda_i)$. Then for an integer $m\ge n$, the matrix $C_g^m$ is NMDS if and only if $\theta^{r} \neq \theta^{r'} $ for all $r,r'\in E$ and $\sum_{i=1}^n \theta^{r_i} = 0$ for some $R = \{r_1,r_2,\ldots,r_n\} \subset E$, where $E=\{0,1,\ldots,n-1,m,m+1,\ldots,m+n-1\}$.
\end{theorem}

\begin{proof}
	We have $\lambda_i = \theta^{i-1}$ for $1\le i\le n-1$ and $\lambda_n = \theta^n$. So for $R = \{r_1,r_2,\ldots,r_{t}\} \subset E$, from (\ref{Eqn_RNMDS_submatrix_G'}), we have

	\begin{equation*}
		\begin{aligned}
			G'[R] &=
			\left[ \begin{array}{cccccccc}
			1 & 1 & \ldots & 1\\
			\theta^{r_1} & \theta^{r_2} & \ldots & \theta^{r_{t}}\\
			\vdots & \vdots & \ddots & \vdots \\
			(\theta^{n-2})^{r_1} & (\theta^{n-2})^{r_2} & \ldots & (\theta^{n-2})^{r_{t}}\\
			(\theta^{n})^{r_1} & (\theta^{n})^{r_2} & \ldots & (\theta^{n})^{r_{t}}\\
			\end{array} \right]
		% 	\\
		% % \end{aligned}
		% % \begin{aligned}
		% 	&
			=
			\left[ \begin{array}{cccccccc}
			1 & 1 & \ldots & 1\\
			\theta^{r_1} & \theta^{r_2} & \ldots & \theta^{r_{t}}\\
			\vdots & \vdots & \ddots & \vdots \\
			(\theta^{r_1})^{n-2} & (\theta^{r_2})^{n-2} & \ldots & (\theta^{r_{t}})^{n-2}\\
			(\theta^{r_1})^{n} & (\theta^{r_2})^{n} & \ldots & (\theta^{r_{t}})^{n}
			\end{array} \right].
		\end{aligned} 
	\end{equation*}

	% Now, to prove the theorem, we can assume $x_{r_i}=\theta^{r_i}$ for $1\leq i \leq t$ and apply Theorem~\ref{Th_GVand_1_NMDS}.

 %    A submatrix $U[R]$, constructed from any $t$ columns of $U$, is given by

	% \begin{equation}\label{Eqn_GVand_MDS_submatrix_U-R}
	% 	U[R] =
	% 	\left[ \begin{array}{cccccccc}
	% 		1 & 1 & \ldots & 1\\ 
	% 		\theta^{r_1} & \theta^{r_2} & \ldots & \theta^{r_t}\\
	% 		(\theta^{r_1})^2 & (\theta^{r_2})^2 & \ldots & (\theta^{r_t})^2\\
	% 		\vdots & \vdots & \ddots & \vdots \\
	% 		(\theta^{r_1})^{n-2} & (\theta^{r_2})^{n-2} & \ldots & (\theta^{r_t})^{n-2}\\
	% 		(\theta^{r_1})^{n} & (\theta^{r_2})^{n} & \ldots & (\theta^{r_t})^{n}
	% 	\end{array} \right],
	% \end{equation}
	% where $R$ denotes a set $\{r_1,r_2,\ldots,r_t\} \subset E = \{1,2,\ldots,2n\}$ of $t$ elements.

	So for $R = \{r_1,r_2,\ldots,r_{n-1}\} \subset E$ we have 

	\begin{equation*}
		\begin{aligned}
			G'[R] &=
			\left[ \begin{array}{cccccccc}
			1 & 1 & \ldots & 1\\
			\theta^{r_1} & \theta^{r_2} & \ldots & \theta^{r_{n-1}}\\
			\vdots & \vdots & \ddots & \vdots \\
			(\theta^{r_1})^{n-2} & (\theta^{r_2})^{n-2} & \ldots & (\theta^{r_{n-1}})^{n-2}\\
			(\theta^{r_1})^{n} & (\theta^{r_2})^{n} & \ldots & (\theta^{r_{n-1}})^{n}
			\end{array} \right].
		\end{aligned} 
	\end{equation*}
	Now, we consider the $(n-1)\times (n-1)$ submatrix $G''[R]$ of $G'[R]$, which is constructed from the first $n-1$ rows of $G'[R]$. Therefore, we have 
	\begin{equation*}
		\begin{aligned}
			G''[R] 
            % &=
			% \left[ \begin{array}{cccccccc}
			% 1 & 1 & \ldots & 1\\
			% \theta^{r_1} & \theta^{r_2} & \ldots & \theta^{r_{n-1}}\\
			% \vdots & \vdots & \ddots & \vdots \\
			% (\theta^{r_1})^{n-3} & (\theta^{r_2})^{n-3} & \ldots & (\theta^{r_{n-1}})^{n-3}\\
			% (\theta^{r_1})^{n-2} & (\theta^{r_2})^{n-2} & \ldots & (\theta^{r_{n-1}})^{n-2}\\
			% \end{array} \right]\\
			&=\vand(\theta^{r_1}, \theta^{r_2}, \ldots, \theta^{r_{n-1}}),
		\end{aligned}
	\end{equation*}
	which is nonsingular since the elements $\theta^{r_1}, \theta^{r_2}, \ldots, \theta^{r_{n-1}}$ are distinct. Therefore, any submatrix of $G'$ constructed from any $n-1$ columns has a nonsingular $(n-1)\times (n-1)$ submatrix, implying that any $n-1$ columns of $G'$ are linearly independent.

	Now, suppose $\sum_{i=1}^n \theta^{r_i'} = 0$ for some $R' = \{r_1',r_2',\ldots,r_n'\} \subset E$. Then for $R'$, we have 
	\begin{equation*}
		\begin{aligned}
			G'[R'] &=
			\left[ \begin{array}{cccccccc}
				1 & 1 & \ldots & 1\\
				\theta^{r_1'} & \theta^{r_2'} & \ldots & \theta^{r_{n}'}\\
				\vdots & \vdots & \ddots & \vdots \\
				(\theta^{r_1'})^{n-2} & (\theta^{r_2'})^{n-2} & \ldots & (\theta^{r_{n'}})^{n-2}\\
				(\theta^{r_1'})^{n} & (\theta^{r_2'})^{n} & \ldots & (\theta^{r_{n}'})^{n}
			\end{array} \right],
		\end{aligned} 
	\end{equation*}
	which is a generalized Vandermonde matrix $V_{\perp}({\bf x};I)$ with ${\bf x}=(\theta^{r_1'}, \theta^{r_2'}, \ldots, \theta^{r_{n}'})$ and $I=\set{n-1}$. Thus, from Corollary~\ref{Corollary_GVand_det1}, we have
	\begin{center}
		$\det(G'[R']) = 
        \left[
          \prod_{1 \le i < j \le n} (\theta^{r_j'} - \theta^{r_i'})
        \right]
        \left(
          \sum_{i=1}^{n} \theta^{r_i'}
        \right)$.
	\end{center}
	Since $\sum_{i=1}^n \theta^{r_i'} = 0$, we have $\det(G'[R'])=0$, i.e., the columns of $G'[R']$ are linearly dependent. Hence, there exist $n$ columns (depending on $R'$) that are linearly dependent.

    % ~\footnote[1]{The argument for the third condition of Lemma~\ref{Lemma_generator_matrix_NMDS} follows the approach suggested by an anonymous reviewer of an earlier version of this work.}

	Now, we need to show that $G'$ also satisfies the third condition of Lemma~\ref{Lemma_generator_matrix_NMDS}. Let $R = \{r_1,r_2,\ldots,r_{n},r_{n+1}\} \subset E$. Then we have 
	\begin{equation*}
		\begin{aligned}
			G'[R] &=
			\left[ \begin{array}{cccccccc}
			1 & 1 & \ldots & 1 & 1\\
			\theta^{r_1} & \theta^{r_2} & \ldots & \theta^{r_{n}} & \theta^{r_{n+1}}\\
			\vdots & \vdots & \ddots & \vdots & \vdots\\
			(\theta^{r_1})^{n-2} & (\theta^{r_2})^{n-2} & \ldots & (\theta^{r_{n}})^{n-2} & (\theta^{r_{n+1}})^{n-2}\\
			(\theta^{r_1})^{n} & (\theta^{r_2})^{n} & \ldots & (\theta^{r_{n}})^{n} & (\theta^{r_{n+1}})^{n}
			\end{array} \right].
		\end{aligned} 
	\end{equation*}
    Observe that the submatrix $G'[R]$ can also be seen as the matrix obtained from the Vandermonde matrix $\vand(\theta^{r_1}, \theta^{r_2}, \ldots, \theta^{r_{n}}, \theta^{r_{n+1}})$ by deleting its $n$-th row (the row corresponding to the power $n-1$). Since the elements $\theta^{r_1},\theta^{r_2}, \ldots, \theta^{r_{n+1}}$ are distinct, the Vandermonde matrix is nonsingular, and hence its $n+1$ rows are linearly independent. Deleting one row from this set of linearly independent rows leaves $n$ linearly independent rows. Hence, $\rank(G'[R])=n$.

    Therefore, by Theorem~\ref{Th_RNMDS_MainResult}, we can conclude that $C_g^m$ is NMDS if and only if $\theta^{r} \neq \theta^{r'} $ for all $r,r'\in E$ and $\sum_{i=1}^n \theta^{r_i} = 0$ for some $R = \{r_1,r_2,\ldots,r_n\} \subset E$.
    \qed
\end{proof}

\begin{example}\label{Example_RNMDS_Cons_I(b)}
	Consider the field $\mathbb{F}_{2^4}$ with the constructing polynomial $x^4+x+1$ and let $\alpha$ be a root of it. Let $\theta=\alpha$. We can verify that $\theta^0+\theta^1+\theta^3+\theta^7=0$. Now, let us consider the polynomial $g(x)=(x-1)(x-\alpha)(x-\alpha^2)(x-\alpha^4)$. It can be verified that $C_g^m$ is an NMDS matrix for $4\leq m \leq 11$.
\end{example}

\begin{remark}\label{Remark_MDS_Construction_I(b)}
	The above theorem assumes that $\sum_{i=1}^n \theta^{r_i} = 0$ for some $R = \{\splitatcommas{r_1,r_2,\ldots,r_n}\} \subset E$. However, to ensure MDS property, the condition needs to be changed to $\sum_{i=1}^n \theta^{r_i} \neq 0$ for all $R = \{r_1,r_2,\ldots,r_n\} \subset E$~\cite[Theorem~3]{GuptaPV19}.
\end{remark}

\begin{lemma}\label{Lemma_RNMDS_clambda}
	If $g(x) = \prod_{i=1}^n (x-\lambda_i) \in \mathbb{F}_q[x]$ yields a recursive MDS (NMDS) matrix then for any $c \in \mathbb{F}_q^{*}$ the polynomial
	$\displaystyle{c^n g\left( \frac{x}{c} \right) = \prod_{i=1}^n (x - c \lambda_i)}$ also yields a recursive MDS (NMDS) matrix.
\end{lemma}

\begin{proof}
	Let $\displaystyle{g^{*}(x) = c^n g\left( \frac{x}{c} \right)}$. Then the matrix $C_{g^{*}} = cEC_gE^{-1}$ where
	\[
	E = \left[
	\begin{array}{cccccc}
	1 ~ & ~ 0 ~ & ~ 0 ~ & ~ \ldots ~ & ~ 0 ~ & ~ 0 \\
	0 ~ & ~ c ~ & ~ 0 ~ & ~ \ldots ~ & ~ 0 ~ & ~ 0 \\
	0 ~ & ~ 0 ~ & ~ c^2 ~ & ~ \ldots ~ & ~ 0 ~ & ~ 0 \\
	&  & & \ldots &  &  \\
	0 ~ & ~ 0 ~ & ~ 0 ~ & ~ \ldots ~ & ~ c^{n-2} ~ & ~ 0 \\
	0 ~ & ~ 0 ~ & ~ 0 ~ & ~ \ldots ~ & ~ 0 ~ & ~ c^{n-1}
	\end{array}
	\right].
	\]
	The matrix $C_{g^{*}}^m = c^mEC_g^mE^{-1}$ is MDS (NMDS) if and only if $C_g^m$ is MDS (NMDS).\qed
\end{proof}

Using the above lemma, it is possible to obtain more polynomials that produce recursive MDS or NMDS matrices from an initial polynomial.

\begin{remark}\label{rem:RNMDS_clambdaI(b)}
	Observe that the condition on $\theta$ in Theorem~\ref{Th_RNMDS_ConsI(b)} is applicable even if we take
	$\lambda_i = \theta^{i-1} c, 1\le i\le n-1,$ and $\lambda_n = \theta^n c$ for some $c\in \mathbb{F}_q^{*}$. 
	By considering the roots in this way, the polynomials that we get are the same as those obtained by applying Lemma \ref{Lemma_RNMDS_clambda}.
\end{remark}

\begin{lemma}\label{Lemma_RNMDS_ConsI(c)}
	Let $\lambda_1 = 1$, and $\lambda_i = \theta^{i},\,2\le i\le n$, for some $\theta\in \mathbb{F}_q^{*}$.
	Let $g(x) = \prod_{i=1}^n (x-\lambda_i)$. Then for an integer $m\ge n$, the matrix $C_g^m$ is NMDS if and only if $\theta^{r} \neq \theta^{r'} $ for all $r,r'\in E$ and $\sum_{i=1}^n \theta^{-r_i} = 0$ for some $R = \{r_1,r_2,\ldots,r_n\} \subset E$, where
	$E=\{0,1,\ldots,n-1,m,m+1,\ldots,m+n-1\}$.
\end{lemma}

\begin{proof}
	Consider $\gamma_i = \lambda_{n-i+1} = (\theta^{-1})^{i-1} c, 1\le i\le n-1$  and
	$\gamma_n = \lambda_{1} = (\theta^{-1})^{n} c$ for $c=\theta^{n}$. Then by Theorem~\ref{Th_RNMDS_ConsI(b)} and the above remark, the matrix $C_g^m$ is NMDS if and only if $\theta^{-r_i},1\le i\le n$, are distinct and $\sum_{i=1}^n \theta^{-r_i} = 0$ for some $R = \{r_1,r_2,\ldots,r_n\} \subset E$.  This completes the proof.\qed
\end{proof}

\begin{example}\label{Example_RNMDS_Cons_I(c)}
	Consider the field $\mathbb{F}_{2^4}$ with the constructing polynomial $x^4+x+1$ and let $\alpha$ be a root of it. Let $\theta=\alpha$. We can verify that $\theta^0+\theta^{-1}+\theta^{-2}+\theta^{-7}=0$. Now, let us consider the polynomial $g(x)=(x-1)(x-\alpha^2)(x-\alpha^3)(x-\alpha^4)$. It can be verified that $C_g^m$ is an NMDS matrix for $4\leq m \leq 11$.
\end{example}

\begin{remark}
	The proof of the above lemma can also be seen similarly as in the proof of Theorem~\ref{Th_RNMDS_ConsI(b)} by using Corollary~\ref{Corollary_GVand_det2}.
\end{remark}

\begin{remark}\label{Remark_MDS_Construction_I(c)}
	The above lemma assumes that $\sum_{i=1}^n \theta^{-r_i} = 0$ for some $R = \{\splitatcommas{r_1,r_2,\ldots,r_n}\} \subset E$. However, to ensure MDS property, the condition needs to be changed to $\sum_{i=1}^n \theta^{-r_i} \neq 0$ for all $R = \{r_1,r_2,\ldots,r_n\} \subset E$~\cite[Corollary~1]{GuptaPV19}.
\end{remark}
Now, we will present a direct construction of polynomial that yields recursive MDS matrix.

\begin{theorem}\label{Th_MDS_New_Construction}
	Let $\lambda_1 = 1$, and $\lambda_i = \theta^{i}$ for $2\le i\le n-1$ and $\lambda_n = \theta^{n+1}$ for some $\theta\in \mathbb{F}_q^{*}$.
	Let $g(x) = \prod_{i=1}^n (x-\lambda_i)$. Then for an integer $m\ge n$, the matrix $C_g^m$ is MDS if and only if $\theta^{r} \neq \theta^{r'} $ for all $r,r'\in E$ and $(\sum_{i=1}^n \theta^{r_i}) (\sum_{i=1}^n \theta^{-r_i}) -1\neq 0$ for all $R = \{r_1,r_2,\ldots,r_n\} \subset E$, where $E=\{0,1,\ldots,n-1,m,m+1,\ldots,m+n-1\}$.
\end{theorem}

\begin{proof}
	We have $\lambda_1 = 1$, and $\lambda_i = \theta^{i}$ for $2\le i\le n-1$ and $\lambda_n = \theta^{n+1}$.
	From Theorem~\ref{Th_RMDS_MainResult}, we know that the matrix $C_g^m$ is MDS if and only if any $n$ columns of $G'$ are linearly independent.
	So for any $R = \{r_1,r_2,\ldots,r_{n}\} \subset E$ we have 
	\begin{equation*}
		\begin{aligned}
			G'[R] &=
			\left[ \begin{array}{cccccccc}
			1 & 1 & \ldots & 1\\
			(\theta^2)^{r_1} & (\theta^2)^{r_2} & \ldots & (\theta^2)^{r_{n}}\\
			\vdots & \vdots & \ddots & \vdots \\
			(\theta^{n-1})^{r_1} & (\theta^{n-1})^{r_2} & \ldots & (\theta^{n-1})^{r_{n}}\\
			(\theta^{n+1})^{r_1} & (\theta^{n+1})^{r_2} & \ldots & (\theta^{n+1})^{r_{n}}\\
			\end{array} \right] 
			% \\
			% &
			=
			\left[ \begin{array}{cccccccc}
			1 & 1 & \ldots & 1\\
			(\theta^{r_1})^2 & (\theta^{r_2})^2 & \ldots & (\theta^{r_{n}})^2\\
			\vdots & \vdots & \ddots & \vdots \\
			(\theta^{r_1})^{n-1} & (\theta^{r_2})^{n-2} & \ldots & (\theta^{r_{n-1}})^{n-2}\\
			(\theta^{r_1})^{n+1} & (\theta^{r_2})^{n+1} & \ldots & (\theta^{r_{n}})^{n+1}
			\end{array} \right].
		\end{aligned} 
	\end{equation*}
	
	Let $y_{r_i}=\theta^{r_i}$ for $1 \leq i \leq n$. Therefore, we have 
	\begin{equation*}
		\begin{aligned}
			G'[R] &=
			\left[ \begin{array}{cccccccc}
				1 & 1 & \ldots & 1\\
				y_{r_1}^2~&~y_{r_2}^2~&~\ldots~&~y_{r_n}^2\\
				\vdots & \vdots & \ddots & \vdots \\
				y_{r_1}^{n-1}~&~y_{r_2}^{n-1}~&~\ldots~&~y_{r_n}^{n-1}\\
				y_{r_1}^{n+1} ~&~ y_{r_2}^{n+1} ~&~ \ldots ~&~ y_{r_n}^{n+1}
			\end{array} \right],
		\end{aligned}
	\end{equation*}
	which is a generalized Vandermonde matrix of the form $V_{\perp}({\bf y}; I)$ with $I=\set{1,n}$. Therefore, from Corollary~\ref{Corollary_GVand_det3} $\det(G'[R])\neq 0$ if and only if $y_{r_i}$ are distinct and $(\sum_{i=1}^n y_{r_i}) (\sum_{i=1}^n y_{r_i}^{-1}) -1\neq 0$. This completes the proof.\qed
\end{proof}

\begin{example}\label{Example_RMDS_Cons_2}
	Consider the field $\mathbb{F}_{2^4}$ with the constructing polynomial $x^4+x+1$ and let $\alpha$ be a root of it. Let $\theta=\alpha$ and consider the polynomial $g(x)=(x-1)(x-\alpha^2)(x-\alpha^3)(x-\alpha^5)$. It can be checked that the polynomial $g(x)$ satisfies the condition in Theorem~\ref{Th_MDS_New_Construction}, so it yields a recursive MDS matrix of order $4$. It can be verified that $C_g^4$ is an MDS matrix.
\end{example}

\noindent So far, we have discussed recursive constructions of MDS and NMDS matrices. In the next section, we will explore the nonrecursive constructions of MDS and NMDS matrices using the direct method.

\section{Direct Construction of Nonrecursive MDS and NMDS Matrices}\label{Section_nonre_Direct_MDS_NMDS}
The application of Vandermonde matrices for constructing MDS codes is well documented in the literature~\cite{MDS_Survey,GR13,LACAN2003,LACAN,HMC1,V_MDS}.
In this section, we explore the use of generalized Vandermonde matrices for the construction of both MDS and NMDS matrices. Specifically, we focus on the generalized Vandermonde matrices $V_{\perp}({\bf x};I)$, where $I$ is a subset of $\set{1, n-1,n}$.

Generalized Vandermonde matrices, with these parameters, defined over a finite field can contain singular submatrices (see Example~\ref{Example_GVand_itself_not_MDS}). Consequently, these matrices by themselves need not be MDS over a finite field. However, like Vandermonde-based constructions, we can use two generalized Vandermonde matrices for constructing MDS matrices.

\begin{example}\label{Example_GVand_itself_not_MDS}
	Consider the generalized Vandermonde matrix $V_{\perp}({\bf x};I)$ with ${\bf x}=(1,\alpha,\alpha^2,\alpha^5)$ and $I=\set{3}$
	\begin{center}
		$V_{\perp}({\bf x};I)=
		\begin{bmatrix}
			1 & 1 & 1 & 1 \\
			1 & \alpha & {\alpha}^2 & \alpha^5\\
			1 & \alpha^2 & \alpha^4 & \alpha^{10}\\
			1 & \alpha^4 & \alpha^{8} & \alpha^{20}\\
		\end{bmatrix}$, 
	\end{center}
	where $\alpha$ is a primitive element of the finite field $\mathbb{F}_{2^4}$ constructed by the polynomial $x^4+x+1$.
	Consider the $2 \times 2$ submatrix 
	\begin{center}
		$\begin{bmatrix}
		1 & \alpha^{5} \\
		1 & \alpha^{20}\\
		\end{bmatrix}$
	\end{center}
	which is singular as $\alpha^{20}=\alpha^{5}$.
\end{example}

\noindent Vandermonde-based MDS matrix constructions may fall within the equivalence class of Cauchy-based constructions~\cite{MDS_Survey}. To move beyond this limitation, the following theorems develop MDS and NMDS constructions using generalized Vandermonde matrices $V_{\perp}({\bf x}; I)$. By omitting selected exponents specified by the set $I$ (for example, $I=\{1\}$ or $I=\{n-1\}$), we fundamentally alter the determinant structure of the matrix. In particular, the determinant is governed not only by products of pairwise differences, but also by expressions involving elementary symmetric polynomials. This additional algebraic flexibility allows us to enforce the precise submatrix rank conditions required for MDS and NMDS codes through appropriate choices of $I$ and suitable nonvanishing sum conditions on the elements $x_i$.

\begin{theorem}\label{Th_GVand_1_MDS}
	Let $V_1 = V_{\perp}({\bf x};I)$ and $V_2 = V_{\perp}({\bf y};I)$ be two generalized Vandermonde matrices with ${\bf x}=(x_1,x_2,\ldots,x_n)$, ${\bf y}=(x_{n+1},x_{n+2},\ldots,x_{2n})$ and $I=\set{n-1}$. The elements $x_i$ are $2n$ distinct elements from $\mathbb{F}_q$, and $\sum_{i=1}^{n}x_{r_i}\neq 0$ for all $R=\set{r_1,r_2,\ldots,r_n}\subset E$, where $E=\set{1,2,\ldots,2n}$. Then the matrices $V_1^{-1}V_2$ and $V_2^{-1}V_1$ are such that any square submatrix of them is nonsingular; hence, they are MDS matrices.
\end{theorem}

\begin{proof}
	Let $U$ be the $n \times 2n$ matrix $[V_1~|~V_2]$. By Corollary~\ref{Corollary_GVand_det1}, we can conclude that both $V_1$ and $V_2$ are nonsingular matrices. Consider the product $G = V_1^{-1}U = [I~|~A]$, where $A = V_1^{-1}V_2$. We will now prove that $A$ does not contain any singular submatrix.

	Now, since $U = V_1G$, from Lemma~\ref{Lemma_AG_another_gen_matrix}, we can say that $U$ is also a generator matrix for the linear code $\mathcal{C}$ generated by matrix $G=[I~|~A]$. From Remark~\ref{Remark_MDS_generator}, we know that a generator matrix $U$ generates a $[2n,n,n+1]$ MDS code if and only if any $n$ columns of $U$ are linearly independent.

	Observe that any $n$ columns of $U$ forms a generalized Vandermonde matrix of the same form as $V_1$ and $V_2$. Since each $x_i$ is distinct and $\sum_{i=1}^{n}x_{r_i}\neq 0$ for all $R=\set{r_1,r_2,\ldots,r_n}\subset E$, from Corollary~\ref{Corollary_GVand_det1}, we can say that every set $n$ column of $U$ is linearly independent. Hence, we can say that the code $\mathcal{C}$ is an MDS code.

    Therefore, $G$ generates a $[2n,n,n+1]$ MDS code and hence $A=V_1^{-1}V_2$ is an MDS matrix. For $V_2^{-1} V_1$, the proof is identical.\qed
\end{proof}

\begin{remark}
    We know that the inverse of an MDS matrix is again MDS \cite{MDS_Survey}; therefore, if $V_1^{-1}V_2$ is MDS, then $V_2^{-1} V_1$ is also MDS and vice versa.
\end{remark}

\begin{remark}
    Note that Theorem~\ref{Th_GVand_1_MDS} provides explicit conditions on the parameters $x_i$: they must be $2n$ distinct elements of $\mathbb{F}_q$, and for every subset $R=\set{r_1,r_2,\ldots,r_n} \subseteq \{1,\dots,2n\}$ of $n$ elements, the sum $\sum_{i=1}^n x_{r_i} \neq 0$. These are the explicit algebraic constraints that guarantee the nonsingularity of the submatrices of $V_1^{-1}V_2$ and $V_2^{-1}V_1$ and hence the MDS property. In practice, such conditions are satisfied by selecting $2n$ distinct elements that avoid zero-sum subsets. For sufficiently large $q$, one can choose elements randomly and then verify the required condition for any $n$-subset $R$. Below, we provide an explicit example over $\mathbb{F}_{2^8}$ yielding a $4 \times 4$ MDS matrix, demonstrating the practicality of these conditions. This approach applies similarly to all constructions presented in the paper, with examples given after each to illustrate their feasibility.
\end{remark}

\begin{example}\label{Example_GVand_1_MDS}
	Consider the generalized Vandermonde matrices $V_1 = V_{\perp}({\bf x};I)$ and $V_2 = V_{\perp}({\bf y};I)$ with ${\bf x}=(1,\alpha,\alpha^2,\alpha^3)$, ${\bf y}=(\alpha^4,\alpha^5,\alpha^6,\alpha^7)$ and $I=\set{3}$, where $\alpha$ is a primitive element of $\FF_{2^8}$ and a root of $x^8+x^7+x^6+x+1$. 
	It can be verified that $V_1$ and $V_2$ satisfy the conditions in Theorem~\ref{Th_GVand_1_MDS}. Therefore, the matrices 
	\begin{equation*}
		\begin{aligned}
			V_1^{-1}V_2&=
			\begin{bmatrix}
				\alpha^{7}   & \alpha^{234} & \alpha^{57} & \alpha^{156}\\
				\alpha^{37}  & \alpha^{66}  & \alpha^{55} & \alpha^{211}\\
				\alpha^{205} & \alpha^{100} & \alpha^{30} & \alpha^{86}\\
				\alpha^{227} & \alpha^{50}  & \alpha^{149} & \alpha^{40}\\
			\end{bmatrix}~\text{and}~V_2^{-1}V_1=
			\begin{bmatrix}
				\alpha^{136} & \alpha^{49} & \alpha^{235} & \alpha^{30}\\
				\alpha^{210} & \alpha^{77} & \alpha^{201} & \alpha^{198}\\
				\alpha^{144} & \alpha^{72} & \alpha^{52}  & \alpha^{220}\\
				\alpha^{42}  & \alpha^{228} & \alpha^{23} & \alpha^{248}
			\end{bmatrix}
		\end{aligned}
	\end{equation*}
	are MDS matrices.
\end{example}
%%%%
\noindent Cauchy matrices are always MDS, meaning that it is not possible to obtain NMDS matrices directly from them. Furthermore, there is currently no known construction method for NMDS matrices using Vandermonde matrices. In Theorem~\ref{Th_GVand_1_NMDS}, we demonstrate the possibility of constructing NMDS matrices using generalized Vandermonde matrices. However, similar to MDS matrices, generalized Vandermonde matrices with $I=\set{n-1}$ themselves may not be NMDS over a finite field (see Example~\ref{Example_GVand_itself_not_NMDS}). As a consequence, we use two generalized Vandermonde matrices for constructing NMDS matrices. 

\begin{example}\label{Example_GVand_itself_not_NMDS}
	Consider the generalized Vandermonde matrix $A=V_{\perp}({\bf x};I)$ with ${\bf x}=(1,\alpha,\alpha^3,\alpha^7)$ and $I=\set{3}$,
	% \begin{center}
	% 	$A=
	% 	\begin{bmatrix}
	% 		1 & 1 & 1 & 1 \\
	% 		1 & \alpha & \alpha^{3} & \alpha^{7} \\
	% 		1 & \alpha^{2} & \alpha^{6} & \alpha^{14} \\
	% 		1 & \alpha^{4} & \alpha^{12} & \alpha^{28}
	% 	\end{bmatrix}$, 
	% \end{center}
	where $\alpha$ is a primitive element of the finite field $\mathbb{F}_{2^4}$ and a root of $x^4+x+1$. Let us consider the linear code $\mathcal{C}$ with a generator matrix
	\begin{equation*}
		\begin{aligned}
			G&=[I~|~A]\\
			 &= \begin{bmatrix}
				1 & 0 & 0 & 0 & 1 & 1 & 1 & 1 \\
				0 & 1 & 0 & 0 & 1 & \alpha & \alpha^{3} & \alpha^{7} \\
				0 & 0 & 1 & 0 & 1 & \alpha^{2} & \alpha^{6} & \alpha^{14} \\
				0 & 0 & 0 & 1 & 1 & \alpha^{4} & \alpha^{12} & \alpha^{28}
			\end{bmatrix}.
		\end{aligned}
	\end{equation*}

\noindent Now, consider matrix
	\begin{equation*}
		\begin{aligned}
			M= \begin{bmatrix}
				0 & 1 & 1 & 1 & 1 \\
				0 & 1 & \alpha & \alpha^{3} & \alpha^{7} \\
				1 & 1 & \alpha^{2} & \alpha^{6} & \alpha^{14} \\
				0 & 1 & \alpha^{4} & \alpha^{12} & \alpha^{28}
			\end{bmatrix},
		\end{aligned}
	\end{equation*}
	which is constructed by the five columns: the third, fifth, sixth, seventh, and eighth columns of $G$. It can be observed that $\rank(M)=3<4$, which violates the condition~$(iii)$ in Lemma~\ref{Lemma_generator_matrix_NMDS}. Therefore, $\mathcal{C}$ is not an NMDS code and hence $A$ is not an NMDS matrix.
\end{example}

\begin{theorem}\label{Th_GVand_1_NMDS}
	Let $V_1 = V_{\perp}({\bf x};I)$ and $V_2 = V_{\perp}({\bf y};I)$ be two generalized Vandermonde matrices with ${\bf x}=(x_1,x_2,\ldots,x_n)$, ${\bf y}=(\splitatcommas{x_{n+1},x_{n+2},\ldots,x_{2n}})$ and $I=\set{n-1}$. The elements $x_i$ are $2n$ distinct elements from $\mathbb{F}_q$ such that $\sum_{i=1}^{n}x_{i}\neq 0$, $\sum_{i=1}^{n}x_{n+i}\neq 0$  and $\sum_{i=1}^{n}x_{r_i}= 0$ for some other $R=\set{r_1,r_2,\ldots,r_n}\subset E$, where $E=\set{1,2,\ldots,2n}$. Then the matrices $V_1^{-1}V_2$  and $V_2^{-1}V_1$ are NMDS matrices.
\end{theorem}

\begin{proof}
	Let $U$ be the $n \times 2n$ matrix $[V_1~|~V_2]$. By Corollary~\ref{Corollary_GVand_det1}, we can conclude that both $V_1$ and $V_2$ are nonsingular matrices. Consider the product $G = V_1^{-1}U = [I~|~A]$, where $A = V_1^{-1}V_2$. To show that $A=V_1^{-1}V_2$ is an NMDS matrix, we need to prove that the $[2n,n]$ code $\mathcal{C}$ generated by $G=[I~|~A]$ is an NMDS code.

	Now, since $U = V_1G$, from Lemma~\ref{Lemma_AG_another_gen_matrix}, we can say that $U$ is also a generator matrix for the linear code $\mathcal{C}$. Thus, we can conclude that $A=V_1^{-1}V_2$ is an NMDS matrix if and only if $U$ meets the three conditions mentioned in Lemma~\ref{Lemma_generator_matrix_NMDS}.

	A submatrix $U[R]$, constructed from any $t$ columns of $U$, is given by
	\begin{equation*}
		U[R] =
		\left[ \begin{array}{cccccccc}
			1 & 1 & \ldots & 1\\ 
			x_{r_1} & x_{r_2} & \ldots & x_{r_t}\\
			x_{r_1}^2 & x_{r_2}^2 & \ldots & x_{r_t}^2\\
			\vdots & \vdots & \ddots & \vdots \\
			x_{r_1}^{n-2} & x_{r_2}^{n-2} & \ldots & x_{r_t}^{n-2}\\
			x_{r_1}^{n} & x_{r_2}^{n} & \ldots & x_{r_t}^{n}
		\end{array} \right],
	\end{equation*}
	where $R$ denotes a set $\{r_1,r_2,\ldots,r_t\} \subset E = \{1,2,\ldots,2n\}$ of $t$ elements.

	Since the $x_{r_i}$ are distinct, the remainder of the proof follows exactly as in Theorem~\ref{Th_RNMDS_ConsI(b)}, with each $\theta^{r_i}$ replaced by $x_{r_i}$. Thus, we conclude that $U$, and hence $G = [I~|~A]$, generates a $[2n,n]$ NMDS code. By Definition~\ref{NMDS_DEF}, it follows that $A = V_1^{-1}V_2$ is an NMDS matrix. The proof for $V_2^{-1} V_1$ is identical.
    \qed
\end{proof}

\begin{remark}\label{Remark_necessity_V1_nonsingular}
	In Theorem~\ref{Th_GVand_1_NMDS}, it is assumed that $\sum_{i=1}^{n}x_{i}\neq 0$ and $\sum_{i=1}^{n}x_{n+i}\neq 0$. This assumption is made based on Corollary~\ref{Corollary_GVand_det1}, which states that $\det(V_{\perp}({\bf x}; I))=\det(\vand({\bf x}))(\sum_{i=1}^{n}{x_i})$ and $\det(V_{\perp}({\bf y}; I))=\det(\vand({\bf y}))(\sum_{i=1}^{n}{x_{n+i}})$. If either of these sums is zero, it would result in the determinant of either $V_1$ or $V_2$ being zero, making them singular. Hence, the assumption is necessary to ensure the nonsingularity of $V_1$ and $V_2$.
\end{remark}

\begin{example}\label{Example_GVand_1_NMDS}
	Consider the generalized Vandermonde matrices $V_1 = V_{\perp}({\bf x};I)$ and $V_2 = V_{\perp}({\bf y};I)$ with ${\bf x}=(1,\alpha,\alpha^2,\alpha^3)$, ${\bf y}=(\alpha^4,\alpha^5,\alpha^6,\alpha^7)$ and $I=\set{3}$, where $\alpha$ is a primitive element of $\FF_{2^4}$ and a root of $x^4+x+1$. It is easy to check that each $x_i$ is distinct and $1+\alpha+\alpha^3+\alpha^7=0$. Therefore, the matrices
	\begin{small}
	\begin{equation*}
		\begin{aligned}
			V_1^{-1}V_2&=
			\begin{bmatrix}
				% \alpha^{3} + \alpha + 1 &~ \alpha^{3} + \alpha &~ \alpha^{3} + \alpha &~ 1 \\
				% \alpha^{3} + 1 &~ \alpha^{3} + 1 &~ \alpha^{3} &~ 1 \\
				% \alpha^{2} + \alpha + 1 &~ \alpha^{2} + \alpha &~ \alpha^{2} + \alpha &~ 0 \\
				% \alpha^{2} &~ \alpha^{2} &~ \alpha^{2} + 1 &~ 1
				\alpha^{7}  ~&~ \alpha^{9}  ~&~ \alpha^{9}  ~&~ 1\\
				\alpha^{14}  ~&~ \alpha^{14}  ~&~ \alpha^{3}  ~&~ 1\\
				\alpha^{10}  ~&~ \alpha^{5}   ~&~ \alpha^{5}  ~&~ 0\\
				\alpha^{2}   ~&~ \alpha^{2}   ~&~ \alpha^{8}  ~&~ 1
			\end{bmatrix}~\text{and}~
			V_2^{-1}V_1=
			\begin{bmatrix}
				% 0 &~ \alpha^{3} + \alpha + 1 &~ 1 &~ \alpha^{3} + \alpha + 1 \\
				% 1 &~ \alpha^{3} + 1 &~ 0 &~ \alpha^{3} \\
				% 1 &~ \alpha^{2} + \alpha &~ 1 &~ \alpha^{2} + \alpha + 1 \\
				% 1 &~ \alpha^{2} + 1 &~ 1 &~ \alpha^{2} + 1
				0  ~&~ \alpha^{7}   ~&~ 1  ~&~ \alpha^{7}\\
				1  ~&~ \alpha^{14}  ~&~ 0  ~&~ \alpha^{3}\\
				1  ~&~ \alpha^{5}   ~&~ 1  ~&~ \alpha^{10}\\
				1  ~&~ \alpha^{8}   ~&~ 1  ~&~ \alpha^{8}
			\end{bmatrix}
		\end{aligned}
	\end{equation*}	
	\end{small}
	are NMDS matrices.
\end{example}

\noindent In the context of implementing block ciphers, we know that if an efficient matrix $M$ used in encryption is involutory, then its inverse $M^{-1}=M$ applied for decryption will also be efficient. Hence, it is important to find MDS or NMDS matrices that are also involutory. 

In the following theorem, we present a method for obtaining involutory matrices from generalized Vandermonde matrices with $I=\set{n-1}$. The proof follows an approach similar to that~\cite[Theorem~4.3]{MDS_Survey} for Vandermonde matrices. However, it is important to note that in the proof of the following theorem, we rely on the conditions $\binom{n}{1}=\binom{n}{n-1}=0$ for even values of $n$ over $\mathbb{F}_{2^r}$. The proof is omitted for brevity.

\begin{theorem}\label{Th_GVand_1_inv_MDS}
	Let $V_1 = V_{\perp}({\bf x};I)$ and $V_2 = V_{\perp}({\bf y};I)$ be two generalized Vandermonde matrices of even order over $\mathbb{F}_{2^r}$ with ${\bf x}=(x_1,x_2,\ldots,x_n)$, ${\bf y}=(y_1,y_2,\ldots,y_n)$ and $I=\set{n-1}$. If $y_{i}=l+x_{i}$ for $i=1,2,\ldots,n$, for some $l\in \mathbb{F}_{2^r}^{*}$ then ${V_2}{V_1}^{-1}$ is a lower triangular matrix whose nonzero elements are determined by powers of $l$. Also, $V_1^{-1}V_2 ~(=V_2^{-1}V_1)$ is an involutory matrix.
\end{theorem}

\begin{remark}
	$V_1^{-1}V_2$ is involutory if and only if $V_1^{-1}V_2=V_2^{-1}V_1$
\end{remark}

\noindent Now, by applying Theorem~\ref{Th_GVand_1_MDS} and Theorem~\ref{Th_GVand_1_inv_MDS}, we can find involutory MDS matrices over $\mathbb{F}_{2^r}$. To force the resulting matrix $V_1^{-1}V_2$ to be involutory, we must establish a symmetric algebraic relationship between the elements of $V_1$ and $V_2$. In the following corollary, we achieve this by applying the strict constraint $x_{n+i} = l + x_i$ for some constant $l$.

\begin{corollary}\label{Corollary_GVand_Inv_MDS}
	Let $V_1 = V_{\perp}({\bf x};I)$ and $V_2 = V_{\perp}({\bf y};I)$ be two generalized Vandermonde matrices of even order over $\mathbb{F}_{2^r}$ with ${\bf x}=(x_1,x_2,\ldots,x_n)$, ${\bf y}=(x_{n+1},x_{n+2},\ldots,x_{2n})$ and $I=\set{n-1}$. If $V_1$ and $V_2$ satisfy the three properties:
	\begin{enumerate}[(i)]
		\item $x_{n+i}=l+x_{i}$ for $i=1,2,\ldots,n$, for some $l\in \mathbb{F}_{2^r}^{*}$,
		\item $x_{i} \neq x_{j}$ for $i\neq j$ where $1 \leq i, j \leq 2n$, and
		\item $\sum_{i=1}^{n}x_{r_i}\neq 0$ for all $R=\set{r_1,r_2,\ldots,r_n}\subset E$, where $E=\set{1,2,\ldots,2n}$,
	\end{enumerate}
	then $V_1^{-1}V_2$ is an involutory MDS matrix.
\end{corollary}

\begin{example}\label{Example_GVand_1_inv_MDS}
	Let $\alpha$ be a primitive element of $\FF_{2^8}$ and a root of $x^8+x^7+x^6+x+1$. Let $l=\alpha$, ${\bf x}=\left( 1, \alpha, \alpha^{2}, \alpha^{3}, \alpha^{4}, \alpha^{5} \right)$, and ${\bf y}=\left( \alpha + 1, 0, \alpha^{2} + \alpha, \alpha^{3} + \alpha, \alpha^{4} + \alpha, \alpha^{5} + \alpha \right)$. Consider the generalized Vandermonde matrices $V_1 = V_{\perp}({\bf x};I)$ and $V_2 = V_{\perp}({\bf y};I)$ with $I=\set{5}$. Then it can be checked that both matrices $V_1$ and $V_2$ satisfy the conditions of Corollary~\ref{Corollary_GVand_Inv_MDS}. Therefore, the matrix
	\begin{align*}
		V_1^{-1}V_2&=
		\begin{bmatrix}
			\alpha^{113} ~&~ \alpha^{33}  ~&~ \alpha^{227} ~&~ \alpha^{93}  ~&~ \alpha^{16} ~&~ \alpha^{174}\\
			\alpha^{63}  ~&~ \alpha^{107} ~&~ \alpha^{186} ~&~ \alpha^{149} ~&~ \alpha^{175} ~&~ \alpha^{10}\\
			\alpha^{105} ~&~ \alpha^{34}  ~&~ \alpha^{116} ~&~ \alpha^{97}  ~&~ \alpha^{198} ~&~ \alpha^{197}\\
			\alpha^{40}  ~&~ \alpha^{66}  ~&~ \alpha^{166} ~&~ \alpha^{43}  ~&~ \alpha^{213} ~&~ \alpha^{52}\\
			\alpha^{136} ~&~ \alpha^{10}  ~&~ \alpha^{185} ~&~ \alpha^{131} ~&~ \alpha^{5}   ~&~ \alpha^{136}\\
			\alpha^{211} ~&~ \alpha^{17}  ~&~ \alpha^{101} ~&~ \alpha^{142} ~&~ \alpha^{53}  ~&~ \alpha^{56}
		\end{bmatrix}
	\end{align*}
	is an involutory MDS matrix.
\end{example}

\begin{remark}
	It is worth mentioning that the above result may not be true for odd order matrices. For example, consider the $3\times 3$ generalized Vandermonde matrices $V_1 = V_{\perp}({\bf x};I)$ and $V_2 = V_{\perp}({\bf y};I)$ with $I=\set{2}$, ${\bf x}=(1, \alpha, \alpha^2)$ and ${\bf y}=(1+\alpha^3, \alpha+\alpha^3, \alpha^2+\alpha^3)$, where $\alpha$ is a primitive element of $\FF_{2^4}$ and a root of $x^4+x+1$. Then it can be checked that the matrices $V_1$ and $V_2$ satisfy the conditions in Corollary~\ref{Corollary_GVand_Inv_MDS}. However, the matrix
	\begin{align*}
		V_1^{-1}V_2&=
		\begin{bmatrix}
			\alpha^{10} ~&~ \alpha^{13} ~&~ \alpha^{1}\\
			\alpha^{3}  ~&~ \alpha^{11} ~&~ \alpha^{11}\\
			\alpha^{11} ~&~ \alpha^{1}  ~&~ \alpha^{13}
		\end{bmatrix}
	\end{align*}
	is not an involutory matrix.
\end{remark}

\noindent By applying Theorem~\ref{Th_GVand_1_NMDS} and Theorem~\ref{Th_GVand_1_inv_MDS}, we can systematically construct involutory NMDS matrices over $\mathbb{F}_{2^r}$ using the following approach. To force the resulting matrix $V_1^{-1}V_2$ to be involutory, we must establish a symmetric algebraic relationship between the elements of $V_1$ and $V_2$. In the following corollary, we achieve this by applying the strict constraint $x_{n+i} = l + x_i$ for some constant $l$. Notably, this work presents the first direct construction method for involutory NMDS matrices over finite fields, providing a concrete framework for generating such matrices rather than relying on exhaustive search.

\begin{corollary}\label{Corollary_GVand_Inv_NMDS}
	Let $V_1 = V_{\perp}({\bf x};I)$ and $V_2 = V_{\perp}({\bf y};I)$ be two generalized Vandermonde matrices of even order over $\mathbb{F}_{2^r}$ with ${\bf x}=(x_1,x_2,\ldots,x_n)$, ${\bf y}=(x_{n+1},x_{n+2},\ldots,x_{2n})$ and $I=\set{n-1}$. If $V_1$ and $V_2$ satisfy the three properties:
	\begin{enumerate}[(i)]
		\item $x_{n+i}=l+x_{i}$ for $i=1,2,\ldots,n$, for some $l\in \mathbb{F}_{2^r}^{*}$,
		\item $x_{i} \neq x_{j}$ for $i\neq j$ where $1 \leq i, j \leq 2n$, and
		\item $\sum_{i=1}^{n}x_{i}\neq 0$, $\sum_{i=1}^{n}x_{n+i}\neq 0$  and $\sum_{i=1}^{n}x_{r_i}= 0$ for some other $R=\set{\splitatcommas{r_1,r_2,\ldots,r_n}}\subset E$, where $E=\set{1,2,\ldots,2n}$,
	\end{enumerate}
	then $V_1^{-1}V_2$ is an involutory NMDS matrix.
\end{corollary}

\begin{example}\label{Example_GVand_1_inv_NMDS}
	Let $\alpha$ be a primitive element of $\FF_{2^4}$ and a root of $x^4+x+1$. Let $l=1$, ${\bf x}=(1,\alpha,\alpha^2,\alpha^3)$, and ${\bf y}=(0,1+\alpha,1+\alpha^2,1+\alpha^3)$. Consider the generalized Vandermonde matrices $V_1 = V_{\perp}({\bf x};I)$ and $V_2 = V_{\perp}({\bf y};I)$ with $I=\set{3}$. Then it can be checked that both matrices $V_1$ and $V_2$ satisfy the conditions of Corollary~\ref{Corollary_GVand_Inv_NMDS}. Therefore, the matrix
	\begin{equation*}
		\begin{aligned}
			V_1^{-1}V_2&=
			\begin{bmatrix}
				\alpha^{9} ~&~ \alpha^{7}  ~&~ \alpha^{7} ~&~ \alpha^{7}\\
				\alpha^{3} ~&~ \alpha^{14} ~&~ \alpha^{3} ~&~ \alpha^{3}\\
				\alpha^{10} ~&~ \alpha^{10} ~&~ \alpha^{5} ~&~ \alpha^{10}\\
				\alpha^{2} ~&~ \alpha^{2}  ~&~ \alpha^{2} ~&~ \alpha^{8}
			\end{bmatrix}
		\end{aligned}
	\end{equation*}
	is an involutory NMDS matrix.
\end{example}
%%%%
\noindent We will now focus on using the generalized Vandermonde matrices $V_{\perp}({\bf x};I)$ with $I=\{1\}$ for constructing MDS and NMDS matrices. Similar to the case of generalized Vandermonde matrices with $I=\{n-1\}$, these matrices alone may not be MDS or NMDS (as shown in Example~\ref{Example_GVand_2_not_MDS_NMDS}). Therefore, we will consider two generalized Vandermonde matrices for the construction of MDS and NMDS matrices.

\begin{example}\label{Example_GVand_2_not_MDS_NMDS}
	Consider the generalized Vandermonde matrix $V_{\perp}({\bf x};I)$ with ${\bf x}=(1,\alpha,\alpha^5,\alpha^{10})$ and $I=\set{1}$
	\begin{center}
		$V_{\perp}({\bf x};I)=
		\begin{bmatrix}
			1 & 1 & 1 & 1 \\
			1 & \alpha^2 & \alpha^{10} & \alpha^{20}\\
			1 & \alpha^3 & \alpha^{15} & \alpha^{30}\\
			1 & \alpha^4 & \alpha^{20} & \alpha^{40}\\
		\end{bmatrix}$, 
	\end{center}
	where $\alpha$ is a primitive element of the finite field $\mathbb{F}_{2^4}$ constructed by the polynomial $x^4+x+1$. But it contains a singular $2\times 2$ submatrix 
	$\begin{bmatrix}
		1 & 1 \\
		\alpha^{15} & \alpha^{30}\\
	\end{bmatrix}$. Hence, $V_{\perp}({\bf x};I)$ is not an MDS matrix. Also, it can be checked that $V_{\perp}({\bf x};I)$ is not an NMDS matrix.
\end{example}
%%%
We can prove the following theorem using Corollary~\ref{Corollary_GVand_det2}, which is similar to the proof of Theorem~\ref{Th_GVand_1_MDS}. The proof is omitted for brevity.

\begin{theorem}\label{Th_GVand_2_MDS}
	Let $V_1 = V_{\perp}({\bf x};I)$ and $V_2 = V_{\perp}({\bf y};I)$ be two generalized Vandermonde matrices with ${\bf x}=(x_1,x_2,\ldots,x_n)$, ${\bf y}=(x_{n+1},x_{n+2},\ldots,x_{2n})$ and $I=\set{1}$. Suppose that the elements $x_i$ are $2n$ distinct nonzero elements from $\mathbb{F}_q$, and $\sum_{i=1}^{n}x_{r_i}^{-1}\neq 0$ for all $R=\set{r_1,r_2,\ldots,r_n}\subset E$, where $E=\set{1,2,\ldots,2n}$. Then the matrices $V_1^{-1}V_2$  and $V_2^{-1}V_1$ are such that any square submatrix of them is nonsingular; hence, they are MDS matrices.
\end{theorem}

\begin{example}\label{Example_GVand_2_MDS}
	Consider the generalized Vandermonde matrices $V_1 = V_{\perp}({\bf x};I)$ and $V_2 = V_{\perp}({\bf y};I)$ with ${\bf x}=(1,\alpha,\alpha^2,\alpha^3)$, ${\bf y}=(\alpha^4,\alpha^5,\alpha^6,\alpha^7)$ and $I=\set{1}$, where $\alpha$ is a primitive element of $\FF_{2^8}$ and a root of $x^8+x^7+x^6+x+1$. It can be verified that $V_1$ and $V_2$ satisfy the conditions in Theorem~\ref{Th_GVand_2_MDS}. Therefore, the matrices
	\begin{equation*}
		\begin{aligned}
			V_1^{-1}V_2&=
			\begin{bmatrix}
				\alpha^{9} & \alpha^{43} & \alpha^{252} & \alpha^{70}\\
				\alpha^{232} & \alpha^{68} & \alpha^{92} & \alpha^{168}\\
				\alpha^{206} & \alpha^{213} & \alpha^{93} & \alpha^{230}\\
				\alpha^{34} & \alpha^{243}  & \alpha^{61} & \alpha^{152}
			\end{bmatrix}~\text{and}~
			V_2^{-1}V_1=
			\begin{bmatrix}
				\alpha^{24} & \alpha^{137} & \alpha^{42} & \alpha^{223}\\
				\alpha^{66} & \alpha^{14} & \alpha^{88} & \alpha^{197}\\
				\alpha^{187} & \alpha^{35} & \alpha^{50} & \alpha^{25}\\
				\alpha^{128} & \alpha^{33} & \alpha^{214} & \alpha^{246}
			\end{bmatrix}
		\end{aligned}
	\end{equation*}
	are MDS matrices.
\end{example}

\noindent In the following theorem, we discuss a new construction of NMDS matrices from the generalized Vandermonde matrices with $I=\set{1}$. The proof can be derived using Corollary~\ref{Corollary_GVand_det2}, following a similar approach to that of Theorem~\ref{Th_GVand_1_NMDS}. We state the result without providing a proof.

\begin{theorem}\label{Th_GVand_2_NMDS}
	Let $V_1 = V_{\perp}({\bf x};I)$ and $V_2 = V_{\perp}({\bf y};I)$ be two generalized Vandermonde matrices with ${\bf x}=(x_1,x_2,\ldots,x_n)$, ${\bf y}=(x_{n+1},x_{n+2},\ldots,x_{2n})$ and $I=\set{1}$. Assume that the elements $x_i$ are $2n$ distinct nonzero elements from $\mathbb{F}_q$ such that $\sum_{i=1}^{n}x_{i}^{-1}\neq 0$, $\sum_{i=1}^{n}x_{n+i}^{-1}\neq 0$ and $\sum_{i=1}^{n}x_{r_i}^{-1}= 0$ for some other $R=\set{r_1,r_2,\ldots,r_n}\subset E$, where $E=\set{1,2,\ldots,2n}$. Then the matrices $V_1^{-1}V_2$  and $V_2^{-1}V_1$ are NMDS matrices.
\end{theorem}

\begin{remark}
	As in Theorem~\ref{Th_GVand_1_NMDS}, by Corollary~\ref{Corollary_GVand_det2}, the assumption $\sum_{i=1}^{n}x_{i}^{-1}\neq 0$ and $\sum_{i=1}^{n}x_{n+i}^{-1}\neq 0$ in Theorem~\ref{Th_GVand_2_NMDS} is necessary to ensure the nonsingularity of $V_1$ and $V_2$.
\end{remark}

\begin{example}\label{Example_GVand_2_NMDS}
	Consider the generalized Vandermonde matrices $V_1 = V_{\perp}({\bf x};I)$ and $V_2 = V_{\perp}({\bf y};I)$ with ${\bf x}=(1,\alpha,\alpha^2,\alpha^3)$, ${\bf y}=(\alpha^4,\alpha^5,\alpha^6,\alpha^7)$ and $I=\set{1}$, where $\alpha$ is a primitive element of $\FF_{2^4}$ and a root of $x^4+x+1$. It is easy to check that each $x_i$ is distinct and $1+\alpha^{-1}+\alpha^{-2}+\alpha^{-7}=0$. Therefore, the matrices
	\begin{equation*}
		\begin{aligned}
			V_1^{-1}V_2&=
			\begin{bmatrix}
				\alpha^{9} & \alpha^{5} & \alpha^{2} & \alpha^{13}\\
				\alpha^{7} & \alpha & \alpha^{10} & \alpha^{9}\\
				\alpha^{11} & 0 & 1 & \alpha^{5}\\
				\alpha^{11} & \alpha^{8} & \alpha^{4} & 0
			\end{bmatrix}~\text{and}~
			V_2^{-1}V_1=
			\begin{bmatrix}
				\alpha^{14} & \alpha^{11} & \alpha^{9} & \alpha^{13}\\
				0 & \alpha^{4} & \alpha^{8} & \alpha^{2}\\
				\alpha^{6} & \alpha^{13} & \alpha^{13} & \alpha^{2}\\
				\alpha^{2} & 1 & \alpha^{4} & \alpha^{6}
			\end{bmatrix}
		\end{aligned}
	\end{equation*}
	are NMDS matrices.
\end{example}

% \begin{remark}
% 	It is important to note that some $x_i$ may be zero in Theorem~\ref{Th_GVand_1_MDS} and Theorem~\ref{Th_GVand_1_NMDS} for $V_1^{-1}V_2$  (or $V_2^{-1}V_1$) to be MDS and NMDS respectively. While in Theorem~\ref{Th_GVand_2_MDS} and Theorem~\ref{Th_GVand_2_NMDS} each $x_i$ need to be nonzero otherwise $V_1$ and $V_2$ will be singular.
% \end{remark}

\noindent Now, we consider generalized Vandermonde matrices $V_{\perp}({\bf x}; I)$, where $I$ has more than one element, specifically, we consider $V_{\perp}({\bf x}; I)$ with $I=\set{1,n}$ for providing a new direct construction for MDS matrices. The proof follows a similar approach to that of Theorem~\ref{Th_GVand_1_MDS} and can be derived using Corollary~\ref{Corollary_GVand_det3}. The proof is omitted for brevity.

\begin{theorem}\label{Th_GVand_3_MDS}
	Let $V_1 = V_{\perp}({\bf x};I)$ and $V_2 = V_{\perp}({\bf y};I)$ be two generalized Vandermonde matrices with ${\bf x}=(x_1,x_2,\ldots,x_n)$, ${\bf y}=(x_{n+1},x_{n+2},\ldots,x_{2n})$ and $I=\set{1,n}$. The elements $x_i$ are $2n$ distinct nonzero elements from $\mathbb{F}_q$, and $(\sum_{i=1}^{n}x_{r_i})(\sum_{i=1}^{n}x_{r_i}^{-1})-1\neq 0$ for all $R=\set{r_1,r_2,\ldots,r_n}\subset E$, where $E=\set{1,2,\ldots,2n}$. Then the matrices $V_1^{-1}V_2$  and $V_2^{-1}V_1$ are such that any square submatrix of them is nonsingular; hence, they are MDS matrices.
\end{theorem}

\begin{example}\label{Example_GVand_3_MDS}
	Consider the generalized Vandermonde matrices $V_1 = V_{\perp}({\bf x};I)$ and $V_2 = V_{\perp}({\bf y};I)$ with ${\bf x}=(1,\alpha,\alpha^2,\alpha^3)$, ${\bf y}=(\alpha^4,\alpha^5,\alpha^6,\alpha^7)$ and $I=\set{1,4}$, where $\alpha$ is a primitive element of $\FF_{2^4}$ and a root of $x^4+x+1$. 
	It can be verified that $V_1$ and $V_2$ satisfy the conditions in Theorem~\ref{Th_GVand_3_MDS}. Therefore, the matrices 
	\begin{equation*}
		\begin{aligned}
			V_1^{-1}V_2&=
			\begin{bmatrix}
				\alpha^{10} & \alpha^{2} & \alpha^{2} & \alpha^{14}\\
				\alpha^{12} & \alpha^{2} & \alpha^{10} & \alpha^{5}\\
				\alpha  & \alpha^{9} & 1 &1\\
				\alpha^{7}  & \alpha^{7} & \alpha^{4} & \alpha^{12}
			\end{bmatrix}~\text{and}~V_2^{-1}V_1=
			\begin{bmatrix}
				\alpha^{7} & \alpha^{4} & \alpha^{12} & \alpha^{2}\\
				\alpha^{5} & \alpha^{10} & \alpha^{9} & \alpha^{6}\\
				\alpha^{5} & 1   & \alpha^{12} & \alpha^{12}\\
				\alpha^{9} & \alpha^{2} & \alpha^{7}  & \alpha^{5}
			\end{bmatrix}
		\end{aligned}
	\end{equation*}
	are MDS matrices.
\end{example}

\begin{remark}
	It is important to note that in Theorem~\ref{Th_GVand_1_MDS} and Theorem~\ref{Th_GVand_1_NMDS}, at most one $x_i$ may be zero for $V_1^{-1}V_2$ and $V_2^{-1}V_1$ to be MDS or NMDS. However, in Theorem~\ref{Th_GVand_2_MDS}, Theorem~\ref{Th_GVand_2_NMDS}, and Theorem~\ref{Th_GVand_3_MDS}, each $x_i$ needs to be nonzero; otherwise, the term $x_i^{-1}$ in the conditions will not be defined.
\end{remark}

\begin{remark}
	We have presented a method for constructing involutory MDS and NMDS matrices using generalized Vandermonde matrices $V_{\perp}({\bf x}; I)$ with $I=\set{n-1}$. However, we have not been able to determine the conditions for constructing involutory MDS and NMDS matrices from generalized Vandermonde matrices with $I=\set{1}$ and $I=\set{1,n}$.
\end{remark}

\begin{remark}
	This paper does not consider generalized Vandermonde matrices $V_{\perp}({\bf x}; I)$ with sets $I$ other than $\set{1}$, $\set{n-1}$, or $\set{1,n}$, or those with size $|I| > 2$. This is because the conditions for being MDS or NMDS matrices become more complicated. However, it is possible to find additional direct constructions of MDS and NMDS matrices by using Theorem~\ref{Th_GVand_discon}.
\end{remark}

\noindent Although we cannot rule out the possibility that our proposed methods may generate an MDS matrix that can also be obtained from a classical construction (e.g., Cauchy or Vandermonde-based constructions), any such overlap would be purely coincidental and would arise solely from the bounded and discrete nature of finite fields. The core novelty of our contribution lies in the underlying algebraic framework used to generate these matrices, which is fundamentally different from existing approaches. In particular, the dependencies in our constructions are not determined solely by simple difference products, as in standard Vandermonde-based constructions, but are also governed by symmetric polynomials, as shown in Corollaries~\ref{Corollary_GVand_det1}, \ref{Corollary_GVand_det2}, and \ref{Corollary_GVand_det3}. Consequently, our theoretical derivation does not naturally reduce to the equivalence class of Cauchy matrices (as shown in \cite[Theorem 5.1]{MDS_Survey}). Therefore, even if a specific matrix produced by our formulas happens to coincide with one obtainable by a known construction, the mathematical machinery we introduce still provides a novel and structurally distinct approach to MDS matrix generation and broadens the available theoretical design space. Table~\ref{tab:literature_comparison} provides a comprehensive comparison with the known literature of direct constructions.

\begin{table}[htpb]
\centering
\caption{Comparison of Proposed Direct Constructions with Existing Literature}
\label{tab:literature_comparison}
\renewcommand{\arraystretch}{1.3} % Adds a bit of padding for readability
\resizebox{\textwidth}{!}{%
\begin{tabular}{|p{3.5cm}|p{4.6cm}|p{5.5cm}|}
\hline
\textbf{Matrix Classification} & \textbf{Existing Literature} & \textbf{Proposed Construction} \\ \hline
\textbf{Nonrecursive MDS} & Cauchy and Vandermonde matrix based constructions, also from a $[2n,n]$ MDS code (For a comprehensive overview on those constructions we refer~\cite{MDS_Survey}) & Generalized Vandermonde matrices $V_{\perp}({\bf x}; I)$ with $I=\set{1}$ or $\set{n-1}$ or $\set{1,n}$ (Theorems~\ref{Th_GVand_1_MDS}, \ref{Th_GVand_2_MDS}, and \ref{Th_GVand_3_MDS}) \\ \hline
\textbf{Nonrecursive NMDS} & From a $[2n,n]$ NMDS code (e.g., \cite{Huang2021,NMDS_code_2022,NMDS_code_2022_2}) & Generalized Vandermonde matrices $V_{\perp}({\bf x}; I)$ with $I=\set{1}$ or $\set{n-1}$ (Theorems~\ref{Th_GVand_1_NMDS} and \ref{Th_GVand_2_NMDS}) \\ \hline
\textbf{Recursive MDS} & Companion matrix based constructions, shortened BCH codes, Gabidulin codes (e.g., \cite{Augot2014,Berger2013,MDS_Survey,GuptaPV17_1,GuptaPV17_2,GuptaPV19}) & Companion matrix based construction (Theorem~\ref{Th_MDS_New_Construction}) \\ \hline
\textbf{Recursive NMDS} & No prior direct construction is known & Companion matrix based construction (Theorem~\ref{Th_RNMDS_ConsI(b)} and Lemma~\ref{Lemma_RNMDS_ConsI(c)}) \\ \hline
\textbf{Involutory MDS} & Cauchy and Vandermonde matrix based constructions (For a comprehensive overview on those constructions we refer~\cite{MDS_Survey}) & Generalized Vandermonde matrices $V_{\perp}({\bf x}; I)$ with $I=\set{1}$ (Corollary~\ref{Corollary_GVand_Inv_MDS}) \\ \hline
\textbf{Involutory NMDS} & No prior direct construction is known & Generalized Vandermonde matrices $V_{\perp}({\bf x}; I)$ with $I=\set{1}$ (Corollary~\ref{Corollary_GVand_Inv_NMDS}) \\ \hline
\end{tabular}%
}
\end{table}

\section{Conclusion}\label{Section_Conclusion_Direct_Cons}
There has been significant research in the literature on the direct construction of MDS matrices using both recursive and nonrecursive methods. However, research on NMDS matrices has been limited in the literature, and there is currently no direct construction method available for them in a recursive approach. This paper addresses this gap by presenting novel direct construction techniques for NMDS matrices in the recursive setting. By employing generalized Vandermonde matrices, we provide a new approach for constructing MDS and NMDS matrices. We also propose a method for constructing involutory MDS and
NMDS matrices using generalized Vandermonde matrices. Moreover, the paper provides proof for some commonly referenced results related to the NMDS codes. Overall, this work provides valuable tools for constructing MDS and NMDS matrices and advances the current state of research in this area. As a promising direction for future work, the theoretical foundations established here can serve as a guide toward efficient implementations. Specifically, applying advanced global optimization heuristics to these newly discovered matrix classes to evaluate concrete implementation metrics, such as area and latency, will be a valuable next step for deploying these structures in lightweight cryptographic primitives.

\medskip
%\bibliographystyle{splncs04}
%\bibliography{Ref}

\end{document}